\documentclass{article}[12]

\usepackage{soul}
\usepackage{amsthm,amsmath,amssymb,epsfig}
\usepackage[mathscr]{eucal}
\usepackage{color}
\usepackage{soul}
\usepackage[table,xcdraw]{xcolor}
\usepackage{comment}
\usepackage{todonotes}
\usepackage{verbatim}
\usepackage{listings}
\usepackage[table]{xcolor}
\newtheorem{theorem}{Theorem}

\newtheorem{example}[theorem]{Example}

\newtheorem{remark}[theorem]{Remark}

\numberwithin{equation}{section} \numberwithin{theorem}{section}

\usepackage{color}
\usepackage[T1]{fontenc}

\newcommand{\blue}{\textcolor[rgb]{0.00,0.00,1.00}}

\title{\textcolor{blue}{Predicting failure times of coherent systems}}
\author{Jorge Navarro}

\author{Jorge Navarro$^{a,}$\thanks{Correspondence to: Jorge Navarro, Facultad
		de Matem\'aticas, Campus de Espinardo, Universidad de Murcia, 30100 Murcia, Spain. E-mail: \textit{jorgenav@um.es}. Telephone number: 34 868883509.
		Fax number: 34 868884182.}, Antonio Arriaza$^{b}$ and Alfonso Su\'arez-Llorens$^{b}$ \\
	$^a$
Universidad de Murcia, Murcia, Spain\\
$^b$Universidad de C\'adiz, C\'adiz, Spain.}

\begin{document}
	\maketitle
	
	\begin{abstract}
    \textcolor{blue}{The paper is focused on studying} how to predict the failure times of coherent systems from the early failure times of their components. Both the cases of independent and dependent components are considered by assuming that they are identically distributed (homogeneous components). The heterogeneous components' case can be addressed similarly but more complexly. The present study is for non-repairable systems, but the information obtained could be used to decide if a maintenance action should be carried out at \blue{time $t$.} Different cases are considered regarding the information available \blue{at time} $t$. We use quantile regression techniques to predict the system failure times and to provide prediction intervals. The theoretical results are applied to specific system structures in some illustrative examples.
	\end{abstract}
	
	
	
\section{Introduction}

Coherent systems are fundamental tools widely used in Reliability Theory. 
\textcolor{blue}{A (binary) system is a map $\phi:\{0,1\}^n\to\{0,1\}$ where $x_1,\dots,x_n$ represent the states of the  components ($x_i=1$ if the component $i$ is working, $x_i=0$ if it has failed) and $\phi(x_1,\dots,x_n)$ represents the state of the system. The map $\phi$ is also known as the structure function of the system. A system consisting of $n$ components is considered coherent if it satisfies two conditions: its structure function is increasing, and all components are relevant. The first condition implies that if the system operates with $m$ components, where $m < n$, then the system cannot transition to a failure state by repairing any of the remaining $n-m$ components. The second condition states that for each component, there exists at least one configuration of states (working/not working) for the remaining components such that the system's state coincides with the state of such a component. The main properties of coherent systems can be seen in  \cite{Alfonso,BP75,N21} and in the references therein.} 

A relevant task in \textcolor{blue}{Reliability Theory} is to define protocols to \blue{predict and} extend the lifetimes of coherent systems. This includes maintenance and redundancy policies. One of these policies is the condition-based maintenance, where coherent systems incorporate sensors in critical components to monitor the system's state. \textcolor{blue}{For example, a system with pipelines can have installed some sensors to measure the temperature and pressure of each pipeline. The deviation of such magnitudes from the normal working conditions can produce early failures of the system. Thus,} the collected information from the sensors can be used to perform predictive maintenance. These techniques are useful to predict when it is suitable to perform a maintenance action on the system and, therefore, to prevent a future failure, \blue{see \cite{AHB23,NAS19,Patward}.} 

Several techniques have been developed to obtain the reliability function and the mean time to failure of a system, \textcolor{blue}{see, for example, \cite{ACB23, AHB23,S07}.} These functions \blue{are} helpful to predict the lifetime of a system and its failure probability. However, generally, it is not easy to predict system failure times from early component failures. This task is especially complex when the components are dependent. Many authors assume independence among the system components \blue{(see, e.g.,  \cite{BP75,E17,S07,YFN04} and the references therein)} but, in several practical situations, it is not a realistic assumption because the components usually share a common environment or load. 

This paper aims to solve this task by providing different tools to get predictions for the system failure time from the information available at a given time $t\geq 0$. \textcolor{blue}{The system lifetime will be represented by the random variable $T$, and the components' lifetimes will be modeled by the random variables $X_1,\dots,X_n$}. Then, it is well known (see, e.g., \cite{N21,NASS16} or Section 2) that the system reliability function $\bar F_T(x)=\Pr(T>x)$ can be obtained as
$$\bar F_T(x)= \bar Q( \bar F_1(x),\dots,\bar F_n(x))$$
for all $x\geq 0$, where $\bar F_i(x)=\Pr(X_i>x)$ for $i=1,\dots, n$ are the components' reliability functions and $\bar Q$ is a distortion function which depends on the system structure and the copula associated to the components' lifetimes. This representation is helpful because, given a system structure with a fixed dependence among the components, we can compute the system reliability function obtained from different kinds of components.

When the system is new, that is, at time $t=0$, the expected system failure time is 
$$E(T)=\int_0^\infty \bar F_T(x) dx =\int_0^\infty \bar Q( \bar F_1(x),\dots,\bar F_n(x)) dx.$$
We will see other options to predict $T$  later.   However, at a given time $t>0$, we have different options to predict the system failure time $T$. These predictions will depend on the information available \blue{at time $t$.} For example,  we may know that the system is working, that is, \textcolor{blue}{the event $A_1=\{T>t\}$ occurs} or we may know that all the components are working, that is, $A_2=\{X_1>t,\dots,X_n>t\}$ \textcolor{blue}{is satisfied}. The residual lifetimes of the system under both assumptions
$(T-t|A_i)$, $i=1,2$ were studied in \cite{N18}.  Another option is to know that some specific components have failed before $t$ and that the others are still working, that is, 
$$A_3=\{X_{i_1}=t_{1},\dots,X_{i_r}=t_r,X_i>t \text{ for } i\notin  \{i_1,\dots,i_r \}\}$$
for $0<t_1<\dots<t_r\leq t$. \blue{This case was partially studied in \cite{ND17}.}

In this paper, we study other situations where we do not know which components failed first and we want to schedule a protocol to be performed when the first component failure occurs \blue{(some systems are equipped with warning alarms at some components or set of components).} In this case, we must study $(T|X_{1:n}=t)$,  where $X_{1:n}=\min(X_1,\dots, X_n)$ represents the first failure of the components. At time $t$, we may know which component fails but this information is not available a priori. Other options are studied as well. For example, we study what happens when we know the second failure time $X_{2:n}$, where $X_{1:n},\dots, X_{n:n}$ represent the ordered failure times of the components  (order statistics). Another option is to assume that we know both failure times, \blue{$X_{1:n}=t_1$ and $X_{2:n}=t$. The results for other failure times are analogous.} In these cases, we may also assume $T>t$ or that $T\geq t$. If the system has already failed \blue{at time} $t$, that is, $T<t$, then the inactivity time of the system  $(t-T|A)$ can be predicted using techniques similar to the ones developed here under different \blue{assumptions.} Some of these cases were studied in \cite{NC19,NPL17}. The framework developed in this manuscript can also predict the lifetimes \blue{(or at least provide a lower bound for a fixed percentage of them)} of coherent systems formed by \blue{modules.  In such a case, we would assume that we know the failure times of the modules instead of the components, see \cite{Torrado21} for a brief introduction to coherent systems composed of modules. This approach can be used to study complex systems with many components.}

To provide such predictions for $T$, under the considered assumptions, we will use quantile regression (QR) techniques that also give prediction bands \blue{for the system failure time. This approach has not been used in the references cited above. It was used in \cite{NB22} to predict $k$-out-of-$n$ systems.} When we only have training data available, we could use the empirical QR techniques to estimate these QR curves or \blue{to estimate} the parameters in the copula and/or in the components' distributions, see \blue{Example \ref{ex1}} and \cite{K05, MN22,Ne06,TLS06}.

The remainder of the paper is structured as follows. Preliminary results and notation are provided in Section \ref{Sec2}. The main tools for the considered assumptions are given in Section \ref{Sec3}. The examples are studied in Section \ref{Sec4}. Finally,  the conclusions and the main tasks for future research projects are expanded in Section \ref{Sec5}.

\section{Notation and preliminary results}\label{Sec2}

In the paper, `increasing' and `decreasing' mean `non-decreasing' and `non-increasing', respectively.  Whenever an expectation, a conditional distribution, or a partial derivative is used, we assume that it exists. The expression $\partial_i \psi$ will denote the partial derivative of a function $\psi$ with respect to its $i$th variable. Analogously, \textcolor{blue}{$\partial_{i,j} \psi$} represents $\partial_i \partial_j \psi$ and so on.

\blue{Let $\phi$ be a coherent system with $n$ components.}
A subset $P$ of $\{1,\dots,n\}$ is a {path set} for $\phi$ if $\phi(x_1,\dots,x_n)=1$ when $x_i=1$ for all $i\in P$, that is, the system functions when all the components in $P$ work.  A path set is a {minimal path set} if it does not contain other path sets. 
The structure function $\phi$ can be written in terms of the minimal path sets as 
\begin{equation}\label{MPS1}
    \phi(x_1,\dots, x_n)=\max_{i=1,\dots,r} \min_{j\in P_i} x_j,
\end{equation}
where $P_1,\dots, P_r$ represent the minimal path sets of the system (see, e.g., \cite{BP75}, p.\ 12). 
By using this representation for the system, its lifetime $T$ can be obtained from the  component lifetimes $X_1,\dots,X_n$ as 
$$T=\phi(X_1,\dots, X_n)=\max_{i=1,\dots,r} \min_{j\in P_i} X_j,$$
where here $\phi:[0,\infty)^n\to[0,\infty)$ is the extension of the system structure function to $[0,\infty)$ obtained from \eqref{MPS1}.

The lifetime of the series system formed with the components in a set $P\subseteq\{1,\dots,n\}$, will be denoted by  $X_P=\min_{j\in P} X_j$. Its reliability (or survival)  function is
\begin{equation}\label{FP}
\bar F_P(t)=\Pr(X_P>t)=\Pr(\cap_{j\in P} \{X_j>t\}).
\end{equation}
Hence, the reliability function of the system $\bar F_T(t)=\Pr(T>t)$ is
\begin{equation}\label{MPS}
	\bar F_T(t)=\sum_{i=1}^r \bar F_{P_i}(t) - \sum_{i=1}^{r-1} \sum_{j=i+1}^{r}  \bar F_{P_i\cup P_j}(t)+\dots +(-1)^{r+1}\bar F_{P_1\cup \dots \cup P_r}(t)
\end{equation}
for all $t\geq 0$.  This expression for the system reliability is called the {minimal path set representation} (see, e.g., \cite{N21}, p.\ 37). Note that it is a linear combination of series system reliability functions. 

The components can be dependent, and this dependence structure will be represented by a {survival copula} $\widehat C$ that is  used to represent their joint reliability function as 
\begin{equation}\label{SC}
\Pr(X_1>x_1,\dots, X_n>x_n)=\widehat C( \bar F_1(x_1),\dots, \bar F_n(x_n))
\end{equation}
(see, e.g., \cite{Ne06}, p. 32, or \cite{DS16}, p. 33), where $\bar F_i(t)=\Pr(X_i>t)$ for $i=1,\dots,n$. The independence case is represented by the product copula 
$$\widehat C(u_1,\dots,u_n)=\Pi (u_1,\dots,u_n)=u_1 \cdots u_n$$ for $u_1, \dots, u_n\in[0,1]$.

\quad

The following notation will be used in the paper. 
For $\mathbf{u}=(u_1,\dots,u_n)$ and $P\subseteq\{1,\dots,n\}$, $\mathbf{u}_P:=(u^P_1,\dots,u^P_n)$ with $u^P_i=u_i$ for $i\in P$ and $u^P_i=1$ for $i\notin P$. For example, if $\mathbf{u}=(u_1,u_2,u_3)$, then 
$\mathbf{u}_{\{1,2\}}=(u_1,u_2,1)$. Then, for a given $n$-dimensional copula $C$,  we define 
$C_P(\mathbf{u}):=C(\mathbf{u}_P)$. Note that this is the copula for the marginal distribution of the random variables included in $P$.  Thus, \blue{from \eqref{FP} and \eqref{SC},}  we can obtain the reliability function of the  series system  $X_P$ as 
$$\bar F_P(t)= \widehat C_P (\bar F_1(t),\dots,\bar F_n(t))$$
for all $t\geq 0$. By using this expression in \eqref{MPS}  the system reliability function can be  written as 
\begin{equation}\label{DR}
\bar F_T(t)=\bar Q (\bar F_1(t),\dots,\bar F_n(t)).
\end{equation}
This expression shows that it is a distortion of the component reliability functions with the distortion function 
$$\bar Q (\mathbf{u})= \sum_{i=1}^r \widehat C_{P_i}(\mathbf{u}) - \sum_{i=1}^{r-1} \sum_{j=i+1}^{r}  \widehat C_{P_i\cup P_j}(\mathbf{u})+\dots +(-1)^{r+1}\widehat C_{P_1\cup \dots \cup P_r}(\mathbf{u}).$$
This distortion  only depends on the structure of the  system  (its minimal path sets) and the dependence structure between the components (their survival copula), see \cite{N21}, p.\  60. 

If the lifetimes of the components  are identically distributed (ID), then \blue{the distortion representation \eqref{DR}} can be reduced to 
\begin{equation}\label{dis1}
\bar F_T(t)=\bar q (\bar F(t))
\end{equation}
for all $t\geq 0$, where $\bar q (u)=\bar Q(u,\dots,u)$ for $u\in[0,1]$ is a univariate distortion and $\bar F$ is the common reliability function of the component lifetimes.

\quad

We need the following additional notation: 
For  $P,P^\ast\subseteq\{1,\dots,n\}$ with $P\cap P^\ast=\emptyset$, $(u,v)_{P,P^\ast}:=(u_1,\dots,u_n)$ with $u_i=u$ for $i\in P$, $u_i=v$ if $i\in P^\ast$ and $u_i=1$ for $i\notin P\cup P^\ast$. Then, for an $n$-dimensional copula $\widehat C$  we define $\widehat C_{P,P^\ast}:[0,1]^2\to [0,1]$ as 
$$\widehat C_{P,P^\ast}(u,v):=\widehat C((u,v)_{P,P^\ast}).$$ 
Therefore, if $X_1,\dots,X_n$ are ID with a reliability function $\bar F$ and a  survival copula $\widehat C$, then
\begin{equation}\label{PP}
\Pr(X_P>x,X_{P^\ast}>y)=\widehat C_{P,P^\ast}(\bar F(x) ,\bar F(y))
\end{equation}
for all $ x, y\geq 0$. 
This representation is a bivariate distortion representation and $ \widehat C_{P,P^\ast}$ is a bivariate distortion function. The multivariate distortion representations were introduced and studied recently in \cite{NCLD21}. They allow to represent the joint distribution function $\mathbf{F}$ of a random vector $(X_1,\dots,X_n)$ as
$$ \mathbf{F}(t_1,\dots,t_n)= D(G_1(t_1),\dots,G_n(t_n))$$
for all $t_1,\dots,t_n$, where $G_1,\dots,G_n$ are some univariate distribution functions and $D:[0,1]^n\to[0,1]$ is a continuous {multivariate distortion function}. The multivariate distortion functions are the restrictions to the set $[0,1]^n$ of multivariate distribution functions with support included in $[0,1]^n$ (see \cite{NCLD21}). A similar representation holds for the joint reliability function with another distortion function. These representations are analogous to the representations based on copulas (and so they have similar properties), but note that here $G_1,\dots,G_n$ are not necessarily equal to the marginal distribution functions of $X_1,\dots,X_n$, and that $D$ is not necessarily a copula function. This happens in \eqref{PP}, where $\bar F$ is not necessarily equal to the reliability functions of the series systems $X_P$ and $X_{P^\ast}$.    

\section{Main results} 
\label{Sec3}

As we have mentioned in the introduction section, we want to predict the failure time of a system $T=\phi(X_1,\dots,X_n)$ from the information available \blue{at time} $t\geq 0$ for some related systems (or structures) $T_j=\phi_j(X_1,\dots,X_n)$ for $j=1,\dots,k$ built from the same components. To this end, we assume that $(X_1,\dots,X_n)$ has an absolutely continuous joint distribution and that the component lifetimes are homogeneous, that is, they have a common reliability function $\bar F$. We consider the following cases of potential practical interest:

\quad

{\bf Case I:} $T_1<T$.

This is the most simple and typical case in practice. Here we have information about an early failure which always happens before the system failure. For example, $T_1$ could be the first component failure, that is, $$T_1=X_{1:n}=\min(X_1, \dots,X_n).$$ 
Observe that we also assume that $T>T_1$ with probability one, that is, \blue{the system does not fail upon the first component failure at time $T_1$.} In this case, as we assume that $(X_1,\dots,X_n)$ has an absolutely continuous joint distribution, then so has $(T_1,T)$.

To perform predictions about $T$, we will use the concept of bivariate distorted distribution, studied in \cite{NCLD21} to represent the joint reliability function of $(T_1,T)$. Then, we will use this representation to get the quantile curves for the conditional random variable $(T|T_1=t)$. The representation for the case of identically distributed (ID) components is obtained in the following theorem.

\begin{theorem}\label{th1}
	If $T_1$ and $T$ are the lifetimes of two coherent systems satisfying $T_1<T$, both systems are based on the same component lifetimes with common reliability function $\bar F$ and $(T_1,T)$ has a joint absolutely continuous distribution, then   there exists a bivariate distortion function $\widehat D:[0,1]^2\to[0,1]$ such that the joint reliability  $\bar G(x,y):= \Pr(T_1>x,T>y)$ can be represented as 
\begin{equation}\label{G1}
\bar G(x,y)=\widehat D(\bar F(x),\bar F(y))
\end{equation}
for all $x,y$. Moreover, the reliability function of $(T|T_1=t)$ is  
\begin{equation}\label{G1b}
\bar G_{T|T_1}(y|t):=\Pr(T>y|T_1=t)=\frac {\partial_1 \widehat D (\bar F(t),\bar F(y)) - \partial_1 \widehat D (\bar F(t),0^+)}{\partial_1 \widehat D (\bar F(t),1) }
\end{equation}
for $y\geq t$, where $\partial_1 \widehat D (u,0^+):=\lim_{v\to 0^+} \partial_1 \widehat D (u,v)$. 
\end{theorem}

\begin{proof}
As we know that $T_1<T$, then
$$\bar G(x,y)= \Pr(T_1>x,T>y) =\Pr(T_1>x)=\bar F_{T_1}(x)=\bar q_{T_1}(\bar F(x))$$
for all $0\leq y\leq x$, where $\bar q_{T_1}$ is the distortion function of the system $T_1$ in representation \eqref{dis1}. 

Let us assume now  $0\leq x<y$ and let $P_1, \dots, P_r$ represent the minimal path sets of the system. Then 
$$\bar G(x,y)= \Pr(T_1>x,T>y)=\Pr\left(T_1>x,\cup_{i=1}^r X_{P_i}>y\right).$$
Hence, by using  the inclusion-exclusion formula we get
\begin{align*}
\bar G(x,y)&= \sum_{i=1}^r \Pr(T_1>x,X_{P_i}>y) - \sum_{i=1}^{r-1} \sum_{j=i+1}^{r}  \Pr(T_1>x,X_{P_i\cup P_j}>y)+\dots \\
&\quad +(-1)^{r+1} \Pr(T_1>x,X_{P_1\cup \dots \cup P_r}>y)
\end{align*}
for $0\leq x<y$. 

Let us assume now that  the minimal path sets of $T_1$ are $P^\ast_1, \dots,P^\ast_s$. Then, for $P\subseteq \{1,\dots,n\}$, by using again the inclusion-exclusion formula we get
\begin{align*}
\Pr(T_1>x,X_{P}>y)
&= \sum_{i=1}^s \Pr(X_{P_i^\ast}>x,X_{P}>y) \\
&\quad- \sum_{i=1}^{s-1} \sum_{j=i+1}^{s}  \Pr(X_{P_i^\ast\cup P_j^\ast }>x,X_{P}>y)+\dots\\ &\quad+(-1)^{s+1} \Pr(X_{P_1^\ast\cup \dots \cup P_s^\ast }>x,X_{P}>y).
\end{align*}
Finally, we note that  for $P,P^\ast\subseteq \{1,\dots,n\}$,
$$\Pr(X_{P^\ast}>x,X_{P}>y)=\Pr(X_{P^\ast- P}>x,X_{P}>y)$$
for all $0\leq x<y$, where $P^\ast-P=P^\ast\cap P^c$ and $P^c$ is the complement  set of $P$ in $\{1,\dots,n\}$. Therefore, as $P^\ast-P$ and $P$ are disjoint sets,  from \eqref{PP}, we get
$$\Pr(X_{P^\ast}>x,X_{P}>y)=\widehat C_{P^\ast,P}(\bar F(x),\bar F(y))$$
which  jointly with the two preceding inclusion-exclusion representations leads to  \eqref{G1} where $\widehat D$ depends on $\widehat C$ and the minimal path sets of both systems.   This representation holds for any reliability function $\bar F$. So $\widehat D$ is a proper bivariate distortion function (see \cite{NCLD21}).

The conditional distribution in \eqref{G1b} can be obtained from \eqref{G1} following the steps of  Proposition 7 in \cite{NCLD21}.  
\end{proof}

\quad

Note that \eqref{G1} is not a copula representation since $\bar F$ is neither the reliability function of $T_1$ nor of $T$. Also, note that the proof of Theorem \ref{th1} shows how to get the distortion function $\widehat D$. In many cases $\partial_1 \widehat D (u,0^+)=0$ holds  and then \eqref{G1b} can be simplified to 
\begin{equation}\label{G1c}
\bar G_{T|T_1}(y|t)=\frac {\partial_1 \widehat D (\bar F(t),\bar F(y))}{\partial_1 \widehat D (\bar F(t),1) }
\end{equation}
which is similar to the expression obtained for copulas. This expression can be used to both compute the conditional expectation as
$$\tilde m(t)=E(T|T_1=t)= \int_0^\infty  \bar G_{T|T_1}(y|t) dy= \int_0^\infty \frac {\partial_1 \widehat D (\bar F(t),\bar F(y))}{\partial_1 \widehat D (\bar F(t),1) }dy,$$
and to get the quantiles of $(T|T_1=t)$. For the latter, we will need the inverse function of $\bar G_{T|T_1}(y|t)$, denoted as $\bar G^{-1}_{T|T_1}(w|t)$, \blue{which is} obtained by solving in $y$ the equation
$$\bar G_{T|T_1}(y|t)=w$$
for $0<w<1$. Note that here we can use analytical or numerical methods  to solve this equation. \blue{Alternatively, we can}  plot the levels curves of the function $\bar G_{T|T_1}(y|t)$ for different values of $w$. 
Thus the curve $m$ for the median regression  is obtained with $w=0.5$ as 
$$m(t)=\bar G^{-1}_{T|T_1}(0.5|t)$$
for $t\geq 0$. Analogously, the centered  prediction band for $T$ at level $90\%$ is obtained with $w=0.05$ and $w=0.95$ as 
$$I_{90}=\left[\bar G^{-1}_{T|T_1}(0.95|t),\bar G^{-1}_{T|T_1}(0.05|t)\right].$$ 
Of course, $\Pr(T\in I_{90}|T_1=t)=0.90$ for all $t\geq 0$. 
Other prediction bands can be obtained similarly. The median regression curve is an excellent alternative to the conditional expectation, and the prediction bands allow us to give more accurate predictions. Examples \ref{ex1} and \ref{ex2} show how to apply this procedure.


\quad

{\bf Case II:} $T_1\leq T$.

This is the most complex case because $(T_1, T)$ has a singular part over the line $T=T_1$. In practice, two options can be considered. In the first one, we are at \blue{time $t>0$} and we know that $T_1=t$ and that $T>t$. Note that if $T=t$, we do not need to predict $T$. In the second case, we are at time zero, and we want to know a priori what will happen when the failure of $T_1$ occurs at \blue{a future time} $t>0$. This case includes when both lifetimes coincide, that is, $T_1=T=t$. Let us see how these cases can be managed.

\quad

{\bf Case II.a:} $T_1=t<T$.

First, we note that the joint reliability function of $(T_1,T)$ can be written as in  \eqref{G1} for this case as well (see the proof of Theorem \ref{th1}). 
Now we might have a singular part in $T=T_1$. However, if the components have an absolutely continuous joint distribution, then the joint distribution of $(T_1,T)$ in the set  $T>T_1$ is absolutely continuous as well. Then \eqref{G1b} holds for $y>t\geq 0$ and can be completed by adding that $\bar G_{T|T_1}(y|t)=1$ for $0\leq y\leq t$. However, note that, in this case
$$\alpha(t):=\Pr(T>t|T_1=t)=\lim_{y\to t^+} \bar G_{T|T_1}(y|t)$$ can be less than $1$. Hence,
\begin{equation}\label{G1bII}
\Pr(T>y|T_1=t<T)=\frac{\Pr(T>y|T_1=t)}{\Pr(T>t|T_1=t)}=\frac {\partial_1 \widehat D (\bar F(t),\bar F(y)) - \partial_1 \widehat D (\bar F(t),0^+)}{\alpha(t)\partial_1 \widehat D (\bar F(t),1) }
\end{equation}
for $y>t$, being one for $0\leq y\leq t$.  This expression can be used to get \blue{the median regression curve $m$}  by solving the equation obtained for the value $0.5$, that is, 
$$\Pr(T>y|T_1=t<T)=0.5.$$
Note that $m(t)>t$ for all $t\geq 0$. 
The prediction bands are obtained similarly.

\quad

{\bf Case II.b:} $T_1=t\leq T$.

This case is actually straightforward, and we can directly use the reliability function given in \eqref{G1b} that now might have a jump at $t$, that is, it might have a mass $\Pr(T=t|T_1=t)=1-\alpha(t)$  at time $t$. In this case, it is better to use bottom prediction bands instead of centered ones. It might also happen that the median regression curve satisfies $m(t)=t$ for some values of $t$.

\quad 

Examples \ref{ex3} and  \ref{ex4} show how to manage cases II.a and II.b.

\quad

{\bf Case III:} $T_1<T_2<T$.

Here the purpose is to use all the information available. We consider a simple case where we know a first failure \blue{at time} $t_1\geq 0$, that is, $T_1=t_1$. Then, we know a second failure $T_2=t_2$ for $t_2\geq t_1$, and we assume $T>T_2$ (with probability one). The other options can be solved similarly (including the case $k>2$).

Proceeding as in the preceding cases, it can be proved that if the components are ID$\sim \bar F$, then the joint reliability of $(T_1,T_2,T)$ can be written as
$$\bar G(t_1,t_2,t)=\widehat D(\bar F(t_1),\bar F(t_2),\bar F(t))$$  
for all $t_1,t_2,t$. Then we assume that this joint reliability is absolutely continuous (e.g. we assume that $\Pr(T_1<T_2<T)=1$) and we obtain its probability density function (PDF) as
 $$g(t_1,t_2,t)=f(t_1)f(t_2)f(t)\partial_{1,2,3}\widehat D(\bar F(t_1),\bar F(t_2),\bar F(t))$$  
 for all $0\leq t_1\leq t_2\leq t$. In a similar manner,  the joint reliability function of $(T_1,T_2)$ can be written as 
 $$\bar G_{1,2}(t_1,t_2)=\bar G(t_1,t_2,0)=\widehat D(\bar F(t_1),\bar F(t_2),1)$$  
 for all $t_1,t_2,t$. Then,  its PDF is
 $$g_{1,2}(t_1,t_2)=f(t_1)f(t_2)\partial_{1,2}\widehat D(\bar F(t_1),\bar F(t_2),1)$$  
 for all $0\leq t_1\leq t_2$. Hence, the PDF of $(T|T_1=t_1,T_2=t_2)$ is 
 $$g_{3|1,2}(t|t_1,t_2)=\frac{g(t_1,t_2,t)}{g_{1,2}(t_1,t_2)}=\frac {\partial_{1,2,3}\widehat D(\bar F(t_1),\bar F(t_2),\bar F(t))}{\partial_{1,2}\widehat D(\bar F(t_1),\bar F(t_2),1)}f(t)$$
for $0\leq t_1\leq t_2\leq t$ such that $f(t_1)f(t_2) \neq0$, and $\partial_{1,2}\widehat D(\bar F(t_1),\bar F(t_2),1)\neq0$. Therefore, the conditional reliability function is 
\begin{equation}\label{GcaseIII}
\bar G_{3|1,2}(t|t_1,t_2)=\frac {\partial_{1,2}\widehat D(\bar F(t_1),\bar F(t_2),\bar F(t))-\partial_{1,2}\widehat D(\bar F(t_1),\bar F(t_2),0^+) }{\partial_{1,2}\widehat D(\bar F(t_1),\bar F(t_2),1)}
\end{equation}
for $0\leq t_1\leq t_2\leq t$ (one for $0\leq t<t_2$). Examples \ref{ex5} and \ref{ex6} show how to use this expression to get quantile regression predictions for $T$ by computing the inverse function of $\bar G_{3|1,2}(t|t_1,t_2)$.

\textcolor{blue}{\begin{remark}
 All cases considered in this paper can be applied to any pair of coherent systems $T_1$ and $T$ since representation \eqref{G1} holds for an appropriate distortion function $\widehat D$. For example, one reviewer suggested studying the case in which the system is equipped with a warning alarm at the second (or the third) component failure, that is, $T_1=X_{2:n}$. This case is studied in Example \ref{ex5} for $k$-out-of-$n$ systems. It is important to note, however, that the procedures proposed here become more difficult as the system's complexity escalates. One alternative to handle these complex systems is to consider groups of components working together as modules. Hence, the system is simplified and the new structure function depends only on the considered modules, see \cite{Torrado21}. Of course, this approach requires knowing the reliability function of each module and the dependence structure among the modules. Furthermore, we would have two different cases. On the one hand, the modules that form the system are identical. For example, a plane might have four engines formed with several components each, and we could study what happens when the engine failure occurs.  On the other hand, if a system is formed by heterogeneous modules (with different reliability functions), then we would only obtain a lower bound of system failure time. Studying new strategies to deal with complex systems is a challenging task for future research projects.  
\end{remark}}

\section{Applications}\label{Sec4}

Let us apply the theoretical results obtained in the preceding section to particular system structures under different assumptions. In the first example, we show how to proceed in a system with IID components under the assumptions of case I.

\begin{example}\label{ex1}
We consider the system with lifetime $$T=\max(X_1,\min(X_2,X_3)).$$ 
Its minimal path sets are $P_1=\{1\}$ and $P_2=\{2,3\}$. Thence, from \eqref{MPS}, the  reliability function of $T$ is
\begin{equation}\label{eq1}
\bar F_T(t)=\Pr(X_1>t)+\Pr(X_{\{2,3\}}>t)-\Pr(X_{\{1,2,3\}}>t)
\end{equation}
for $t\geq 0$ ($1$ elsewhere). If the components are IID and $\bar F$ is their common reliability function, then 
$$\bar F_T(t)=\bar F(t)+\bar F^2(t)-\bar F^3(t)=\bar q(\bar F(t))$$
for $t\geq 0$,  where $\bar q(u)=u+u^2-u^3$ for $u\in[0,1]$. Then the system expected lifetime is 
$$E(T)=\int_0^\infty \bar q(\bar F(t)) dy.$$
For example, if $\bar F(t)=\exp(-t/\mu)$ for $t\geq 0$ (exponential distribution with mean $\mu$), then $E(T)= 7\mu/6=1.166667\mu$. This is the prediction (expected value) at time $t=0$; the system is slightly better than a system with a single component.

Now let us predict the residual lifetime of the system at the first component failure, that is, let us consider $T_1=X_{1:3}=t$ for $t\geq 0$. By using the procedure showed in Theorem \ref{th1}, the joint reliability function of $(T_1,T)$ is
$$\bar G(x,y)=\Pr(T_1>x,T>y)=\Pr(T_1>x)=\bar F^3(x)$$
for $0\leq y<x$ and
\begin{align*}
\bar G&(x,y)=\\
&=\Pr(X_{1:3}>x,X_1>y)+\Pr(X_{1:3}>x, X_{\{2,3\}}>y)-\Pr(X_{1:3}>x,X_{1:3}>y)\\
&=\Pr(X_1>y,X_2>x,X_3>x)+\Pr(X_{1}>x, X_{2}>y,X_{3}>y)-\Pr(X_{1:3}>y)\\
&=\bar F^2(x)\bar F(y)+\bar F(x)\bar F^2(y)-\bar F^3(y)
\end{align*}
for $0\leq x\leq y$. Hence $\bar G(x,y)=\widehat D(\bar F(x),\bar F(y))$ for all $x,y$, where
$$\widehat D(u,v)=\left\{\begin{array}{crr}
u^3 & \text{for}&0\leq u < v\leq 1;\\
u^2v+uv^2-v^3 & \text{for}&0\leq v \leq  u\leq 1;\\
\end{array}%
\right.$$
and
$$\partial_1 \widehat D(u,v)=\left\{\begin{array}{crr}
3u^2 & \text{for}&0\leq u < v\leq 1;\\
2uv+v^2& \text{for}&0\leq v \leq  u\leq 1.\\
\end{array}%
\right.$$ 
Note that $\lim_{v\to 0^+}\partial_1 \widehat D(u,v)=0$. Then, from \eqref{G1c}, we get
$$\bar G_{T|T_1}(y|t)=\Pr(T>y|T_1=t)=\frac {\partial_1 \widehat D (\bar F(t),\bar F(y))}{\partial_1 \widehat D (\bar F(t),1) }=\frac {2\bar F(y)\bar F(t)+\bar F^2(y)}{3\bar F^2(t) }$$
for $0\leq t\leq y$ ($1$ for $y\leq t$). Note that it is a mixture of the residual lifetime of a single component with weight $2/3$ ($X_2$ or $X_3$ are the first failure) and the residual lifetime of a series system with two IID components with weight $1/3$ ($X_1$ is the first failure).

The associated inverse function for $0<w<1$ is then obtained by solving the quadratic equation 
$$\bar F^2(y)+2\bar F(t)\bar F(y)-3w \bar F^2(t)=0,$$
obtaining
$$\bar G^{-1}_{T|T_1}(w|t)=\bar F^{-1}\left(-\bar F(t)+\bar F(t)\sqrt{1+3w}\right)$$
that is the unique positive solution for $0<w<1$ and $t\geq 0$. Therefore, we can predict $T$ by using the median regression curve
$$m(t)=\bar G^{-1}_{T|T_1}(0.5|t)=\bar F^{-1}\left(\bar F(t)\left(\sqrt{2.5}-1\right)\right)\approx \bar F^{-1}\left(0.5811388\bar F(t)\right)$$
for $t\geq 0$. The centered $90\%$ prediction band for $T$ is 
\begin{equation}
\label{90bandEx41}
I_{90}(t)=\left[\bar G^{-1}_{T|T_1}(0.95|t),\bar G^{-1}_{T|T_1}(0.05|t)\right],   
\end{equation}
that is,
$$I_{90}(t)=\left[\bar F^{-1}\left(\bar F(t)\left(\sqrt{3.85}-1\right)\right),\bar F^{-1}\left(\bar F(t)\left(\sqrt{1.15}-1\right)\right) \right].$$
Analogously, the centered $50\%$ prediction band for $T$ is obtained with 
\begin{equation}
\label{50bandEx41}
I_{50}(t)=\left[\bar G^{-1}_{T|T_1}(0.75|t),\bar G^{-1}_{T|T_1}(0.25|t)  \right].
\end{equation}
If the \textcolor{blue}{components}  have an  exponential distribution with mean $\mu$, then  
\begin{equation}
\label{medianEx41}
m(t)=t-\mu \log \left(\sqrt{2.5}-1\right)=t+0.5427656\mu,    
\end{equation}
and the mean regression curve is 
\begin{equation*}
\tilde m(t)=E(T|T_1=t)=\int_0^\infty  \bar G_{T|T_1}(y|t)dy=t+\frac 5 6 \mu=t+0.8333333 \mu
\end{equation*}
for $t\geq 0$. The quantile regression curves are also straight lines. As expected from the independence assumption and the lack of memory property of the exponential distribution, the predictions for the residual lifetime $(T-t|T_1=t)$ do not depend on $t$.

In Figure \ref{fig1}, left,  we provide the plots of the median \blue{(red-dot-dashed line)} and mean \blue{(black-dot-dashed line)} regression curves and the prediction bands for a standard exponential distribution jointly with a scatterplot of a simulated sample from $(T_1,T)$ of size $100$. \blue{In Figure \ref{fig1}, right,} we estimate these curves (lines) by using linear quantile regression (LQR) (for $m$ and the prediction band limits) and linear regression (for $\tilde m$). The basic theory for LQR can be seen in Koenker \cite{K05}. 

Note that the prediction bands explain better the uncertainty in these predictions than the single mean or median regression curves. For example, the first data in our sample is $T_1=0.4632196$ and $T=0.8434573$. The predictions for $T$ at this failure time for $T_1$ are $m(T_1)=1.105407$ and $\tilde m(T_1)=1.296553$, which are quite far from the exact value. However, the centered prediction intervals for this value are
$I_{50}(T_1)=[0.8554071,1.355407]$  and $I_{90}(T_1)=[0.6554071,1.555407]$. The first one does not contain the exact value (it is close to the left margin) but the second does. 

This sort of information is important to decide if we should perform repairs in the system at time $T_1$, which is when the first component failure occurs. 
Even more, note that we can choose the desired prediction bands. For example, one could prefer to choose the bottom $90\%$ prediction band for $T$ $$I^{bottom}_{90}(t)=\left[t,\bar G^{-1}_{T|T_1}(0.10|t)  \right]=\left[t,\bar F^{-1}\left(\sqrt{1.3}\bar F(t)-\bar F(t)\right) \right].$$ 
Also, note that the estimations obtained by using LQR are good, except the one for the $0.95$ regression curve, which has a strong dependency on extreme data. However, we must note that,  in practice,  we will not know if the regression curves are straight lines since we will not know the underlying model for the components.


\begin{table}
{\rowcolors{2}{gray!20!white!50}{gray!70!white!30}
\begin{center}
\begin{tabular}{ |p{0.85cm}|p{1.65cm}|p{2.5cm}|p{2.5cm}|}
\hline
\hspace{0.25cm}$k$ & Replications & Coverage probabilities of $\hat{I}_{50}(t)$ & Coverage probabilities of $\hat{I}_{90}(t)$\\
\hline
\hspace{0.15cm} 1     &\hspace{0.25cm} 1000 &\hspace{0.6cm} 0,36327 &\hspace{0.6cm} 0,71278 \\
\hspace{0.15cm} 5     &\hspace{0.25cm} 1000 & \hspace{0.6cm} 0,46193 &\hspace{0.6cm} 0,85889\\
\hspace{0.15cm}10     &\hspace{0.25cm} 1000 & \hspace{0.6cm} 0,48125 &\hspace{0.6cm} 0,87922\\
\hspace{0.15cm}25     &\hspace{0.25cm} 1000 & \hspace{0.6cm} 0,49396 &\hspace{0.6cm} 0,89131 \\    
\hspace{0.15cm}50     &\hspace{0.25cm} 1000 & \hspace{0.6cm} 0,49748 &\hspace{0.6cm} 0,89591 \\    
 100    &\hspace{0.25cm} 1000 & \hspace{0.6cm} 0,49877 &\hspace{0.6cm} 0,89739 \\    

\hline
\end{tabular}
\end{center}}
\caption{\textcolor{blue}{Coverage probabilities of the estimated centered $50\%$ prediction band $(\hat{I}_{50}(t))$ and centered $90\%$ prediction band $(\hat{I}_{90}(t))$ depending on the sample size $k$ used to calculate $\hat{\mu}$.}}
\label{Tabla01}  
\end{table}

\blue{Sometimes the parameter in the model is unknown. This is the case when the reliability of the components measured in lab could have a different performance when they are installed in the system. From a realistic point of view, we think that we should estimate it from the values of $T_1$. For example, we can assume IID components with exponential distributions and an unknown common mean $\mu$.  Then, to estimate $\mu$, we can consider two situations.  In the first one,  we have several systems with the same structure and a common mean $\mu$. Then, as $E(T_1)=\mu/3$, we can estimate $\mu$ with $\widehat \mu=3\bar T_1$, where $\bar T_1$ is the sample mean for the values of $T_1$. The median regression curve and the prediction bands are then obtained by replacing in \eqref{medianEx41} and \eqref{90bandEx41} or \eqref{50bandEx41}, the unknown value $\mu$ with $\widehat \mu$.  We have performed a simulation study to determine the coverage probabilities of the prediction intervals obtained by using this procedure. 
We consider a sample of $k$ systems and we estimate $\mu$ from $k$ values of $T_1$. Then, we calculate the percentage of values of $T$ that belong to the $50\%$ and $90\%$ centered prediction intervals, obtained from $T_1$ and using the estimate of $\mu$, $\widehat \mu$. We repeat this experiment $1000$ times and calculate the average of the corresponding percentages. The results obtained varying $k$ (the number of systems considered to calculate the estimated value $\hat{\mu}$) can be seen in Table \ref{Tabla01}. For $k = 1$, the coverage probabilities of the estimated prediction intervals are $36.327\%$ and $71.278\%$. Clearly, these values are below the expected ones because the estimation of $\mu$ is poor (it is based on only one data). However, the results improve just by taking $k=5$, which is a really small sample size to estimate $\mu$.}

\blue{Finally, note that if we also have information about the broken component at time $t$, then we can get better predictions (see \cite{NAS19}). For example, if we know that $T_1=X_{1:3}=X_2=t$, then 
$$\Pr(T>y|T_1=X_2=t)=\Pr(X_1>y|X_1>t)=\frac{\bar F(y)}{\bar F(t)}$$
for $y\geq t$. The expression for $T_1=X_{1:3}=X_3=t$ is the same. However, for $T_1=X_{1:3}=X_1=t$ we have
$$\Pr(T>y|T_1=X_1=t)=\Pr(\min(X_2,X_3)>y|X_2>t,X_3>t)=\frac{\bar F^2(y)}{\bar F^2(t)}$$
for $y\geq t$. The median regression curves and the prediction bands for these cases can be obtained easily from these two expressions.} 
\end{example}

\begin{figure}
	\begin{center}
	\includegraphics*[scale=0.338]{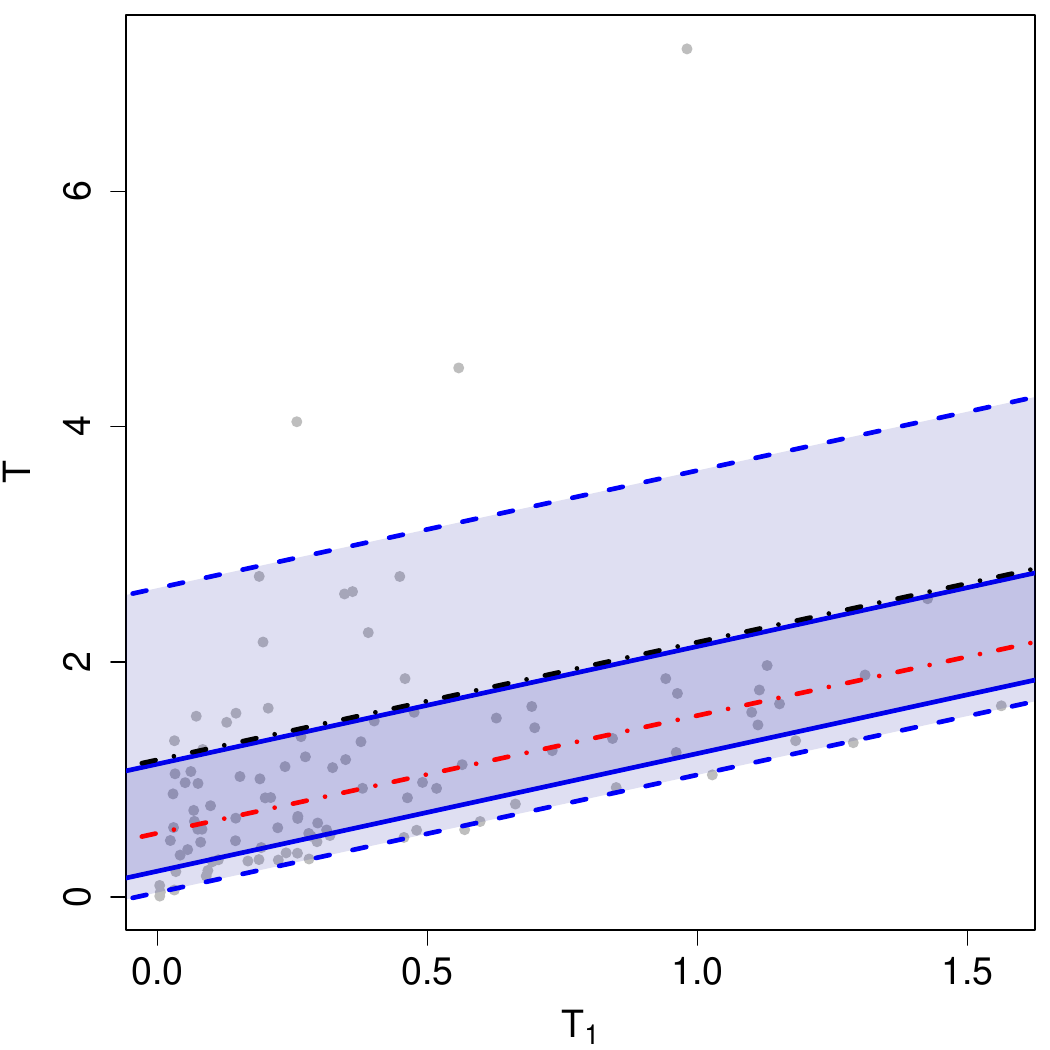}	
    \includegraphics*[scale=0.338]{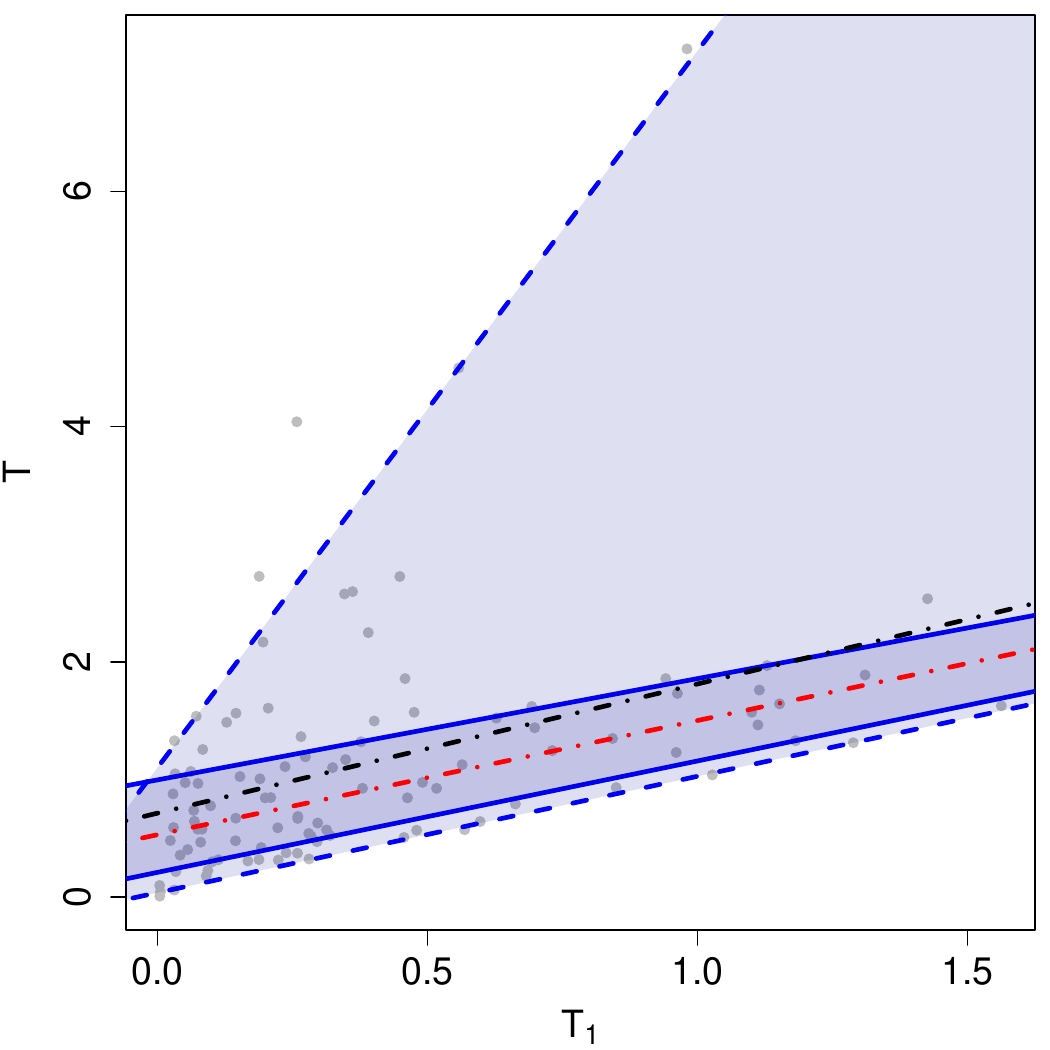}
		\caption{Scatterplots of a sample from $(T_1,T)$ for the systems in Example \ref{ex1} jointly with the theoretical (left) and estimated (right) median \textcolor{blue}{(red-dot-dashed line)} and mean \textcolor{blue}{(black-dot-dashed line)} regression curves  and prediction bands with confidence levels $50\%$ (dark grey) and $90\%$ (light grey).} \label{fig1}
	\end{center}
\end{figure}


In the following example, we introduce a positive dependence between some components to see how it affects the predictions. We use the system structure of Example \ref{ex1} and we assume that the components are ID as well (thus we can compare both cases).

\begin{example}\label{ex2}
We consider again the system $T=\max(X_1,\min(X_2,X_3))$ but now, as the second and third components are in the same path, we assume that they are affected by a common environment (or load), and so they are dependent. To represent this dependency, we use a Clayton survival copula with $\theta=1$  for them (see, e.g., line 1 of Table 4.1 in \cite{Ne06}, p. 116), which leads to a positive dependence. As the first component is in a separate path,  we assume that it is independent of the other components. Hence, if the components are ID, the joint reliability function of the component lifetimes is 
$$\Pr(X_1>x_1,X_2>x_2,X_3>x_3)=  \widehat C(\bar F(x_1), \bar F(x_2),\bar F(x_3))$$
for $x_1,x_2,x_3\geq 0$, where 
$$\widehat C(u_1,u_2,u_3)= \frac {u_1u_2u_3}{u_2+u_3-u_2u_3}$$
for $u_1,u_2,u_3\in[0,1]$. Hence, from \eqref{eq1}, the system reliability is 
$$\bar F_T(t)=\bar F(t)+\frac{\bar F(t)}{2-\bar F(t)}-\frac{\bar F^2(t)}{2-\bar F(t)} $$
for $t\geq 0$. 

If the component lifetimes have an exponential distribution with mean $\mu$, numerically we get that the prediction for $T$ at time $t=0$ is $E(T)=1.306853\mu$,  which is slightly larger than the expected result obtained in the preceding example with independent components ($1.166667\mu$). This is an expectable result since the positive dependence between components $2$ and $3$ improves the series system in the minimal path set formed with these components.

To obtain the quantile regression curve, we need to get the joint reliability function of $(T_1,T)$, where $T_1=X_{1:3}=\min(X_1,X_2,X_3)$. 
$$\bar G(x,y)=\Pr(T_1>x,T>y)=\Pr(T_1>x)=
\frac {\bar F^2(x)}{2-\bar F(x)}
$$
for $0\leq y<x$, and
\begin{align*}
\bar G(x,y)
&=\Pr(X_1>y,X_2>x,X_3>x)+\Pr(X_{1}>x, X_{2}>y,X_{3}>y)\\
&\quad -\Pr(X_1>y,X_2>y,X_3>y)\\
&=\frac{\bar F(x)\bar F(y)}{2-\bar F(x)}+\frac{\bar F(x)\bar F(y)}{2-\bar F(y)}-\frac{\bar F^2(y)}{2-\bar F(y)}
\end{align*}
for $0\leq x\leq y$. Hence, $\bar G(x,y)=\widehat D(\bar F(x),\bar F(y))$ for all $x,y$, where
$$\widehat D(u,v)=\left\{\begin{array}{crr}
\frac{u^2}{1-u} & \text{for}&0\leq u < v\leq 1;\\
\frac{uv}{1-u}+\frac{uv-v^2}{1-v}  & \text{for}&0\leq v \leq  u\leq 1;\\
\end{array}%
\right.$$
and
$$\partial_1 \widehat D(u,v)=\left\{\begin{array}{crr}
 \frac{2u-u^2}{(1-u)^2} & \text{for}&0\leq u < v\leq 1;\\
\frac {v}{(1-u)^2}+\frac {v}{1-v} & \text{for}&0\leq v \leq  u\leq 1.\\
\end{array}%
\right.$$ 
Note that $\lim_{v\to 0^+}\partial_1 \widehat D(u,v)=0$. Then, from \eqref{G1c}, we get
$$\bar G_{T|T_1}(y|t)=\frac {\partial_1 \widehat D (\bar F(t),\bar F(y))}{\partial_1 \widehat D (\bar F(t),1) }=\frac {F(y)\bar F(y)+\bar F(y)F^2(t)}{\bar F(t) F(y) (2-\bar F(t)) }$$
for $y>t$ (one  for $0\leq y\leq t$). The associated quantile function for $0<w<1$ is obtained by solving (in $y$) the quadratic equation 
$$\bar F^2(y)-\left[1+F^2(t)+w\bar F(t)(2-\bar F(t))  \right] \bar F(y)+w\bar F(t) (2-\bar F(t)) =0,$$
obtaining
$$\bar G^{-1}_{T|T_1}(w|t)=\bar F^{-1}\left(\frac{b(t,w)-\sqrt{b^2(t,w)-4c(t,w)}}{2}\right),$$
where $b(t,w)= 1+F^2(t)+w\bar F(t)(2-\bar F(t))$ and $c(t,w)=w\bar F(t) (2-\bar F(t))$. It is the unique solution of the quadratic equation for $\bar F(y)\in[0,1]$ when $0<w<1$ and $t\geq 0$. Therefore, we can predict $T$ by using the median quantile regression curve
$$m(t)=\bar G^{-1}_{T|T_1}(0.5|t)=\bar F^{-1}\left(\frac{b(t,0.5)-\sqrt{b^2(t,0.5)-4c(t,0.5)}}{2}\right)$$
for $t\geq 0$. Analogously, the centered $90\%$ and  $50\%$ prediction  bands for $T$ are $I_{90}(t)=\left[\bar G^{-1}_{T|T_1}(0.95|t),\bar G^{-1}_{T|T_1}(0.05|t)  \right]$
and 
$I_{50}(t)=\left[\bar G^{-1}_{T|T_1}(0.75|t),\bar G^{-1}_{T|T_1}(0.25|t)  \right]$. 
The mean regression curve is 
$$\tilde m(t)=\int_0^\infty  \bar G_{T|T_1}(y|t)dy=t+\int_t^\infty \frac {F(y)\bar F(y)+\bar F(y)F^2(t)}{\bar F(t) F(y) (2-\bar F(t))} dy$$
for $t\geq 0$. Hence, we do not have an explicit expression for it. In the case of ID components with standard exponential distribution, we get 
\begin{align*}
\tilde m(t)&=t+\int_t^\infty \frac {\bar F(y)}{\bar F(t) (2-\bar F(t))} dy
+\int_t^\infty \frac {\bar F(y)F^2(t)}{\bar F(t) F(y) (2-\bar F(t))} dy\\
&=t+\frac {1}{e^{-t} (2-e^{-t})} \int_t^\infty e^{-y}dy
+\frac {(1-e^{-t})^2}{e^{-t} (2-e^{-t})}\int_t^\infty \frac {e^{-y}}{1-e^{-y}} dy\\
&=t+\frac {1}{2-e^{-t}} 
-\frac {(1-e^{-t})^2}{e^{-t} (2-e^{-t})}\log(1-e^{-t})
\end{align*}
for $t\geq 0$. Note that  $\tilde m(0)=0.5< E(T)=1.306853$ as expected. It is also worse than the expected value under independent components ($5/6=0.8333333$) since the positive dependency and an early failure \blue{at time} $t=0$ for $X_2$ (or $X_3$) leads to a close value to zero for $X_3$ (or $X_2$).

In Figure \ref{fig2}, left,  we plot the median \blue{(red-dot-dashed line)} and mean \blue{(black-dot-dashed line)} regression curves and these prediction bands for a standard exponential distribution jointly with a scatterplot of a simulated sample from $(T_1,T)$ of size $100$. To get this sample we use the inverse transform method described for example in \cite{MS02}, p.\ 88. Note that the quantile curves are almost straight lines. Indeed, they are very similar to that obtained in the independent case. We observe that only $6$ data points (failure times) are out of the $90\%$ centered prediction band ($2$ above and $4$ below). Hence, these prediction bands give accurate predictions considering the uncertainty in this procedure.

In Figure \ref{fig2}, right, we estimate these curves by using linear quantile regression  (for $m$ and the limit of the prediction bands) and linear regression (for $\tilde m$). In this case, the worst estimation is obtained for the upper bound of the $50\%$ centered prediction band (top continuous blue line).

As in the preceding example,  the prediction bands explain better the uncertainty in these predictions than the single mean or median regression curves. For example, the first data in our sample is
$T_1=0.04599828$ and $T=0.3444294$. The predictions for $T$ at this failure time for $T_1$ with the exact median and mean regression curves are   $m(T_1)=0.6991802$ and $\tilde m(T_1)=1.009258$, which are not close to the real value. 
The exact centered prediction  intervals for $T$ are
$I_{50}(T_1)=[0.297528, 1.391004]$  and $I_{90}(T_1)=[0.07943828,2.99988]$. Both intervals contain the exact value. In general, from Figure \ref{fig2}, left,  we know that the $95\%$ of the systems will fail between $0$ and $3$ units of ``times'' (years, moths, cycles, etc.) after the first component failure. This information is important to decide if we should perform repairs or replacements in the system at time $T_1$. In other situations (systems, copulas and distributions), these predictions may depend more on $t$.
\end{example}

\begin{figure}
	\begin{center}
		\includegraphics*[scale=0.338]{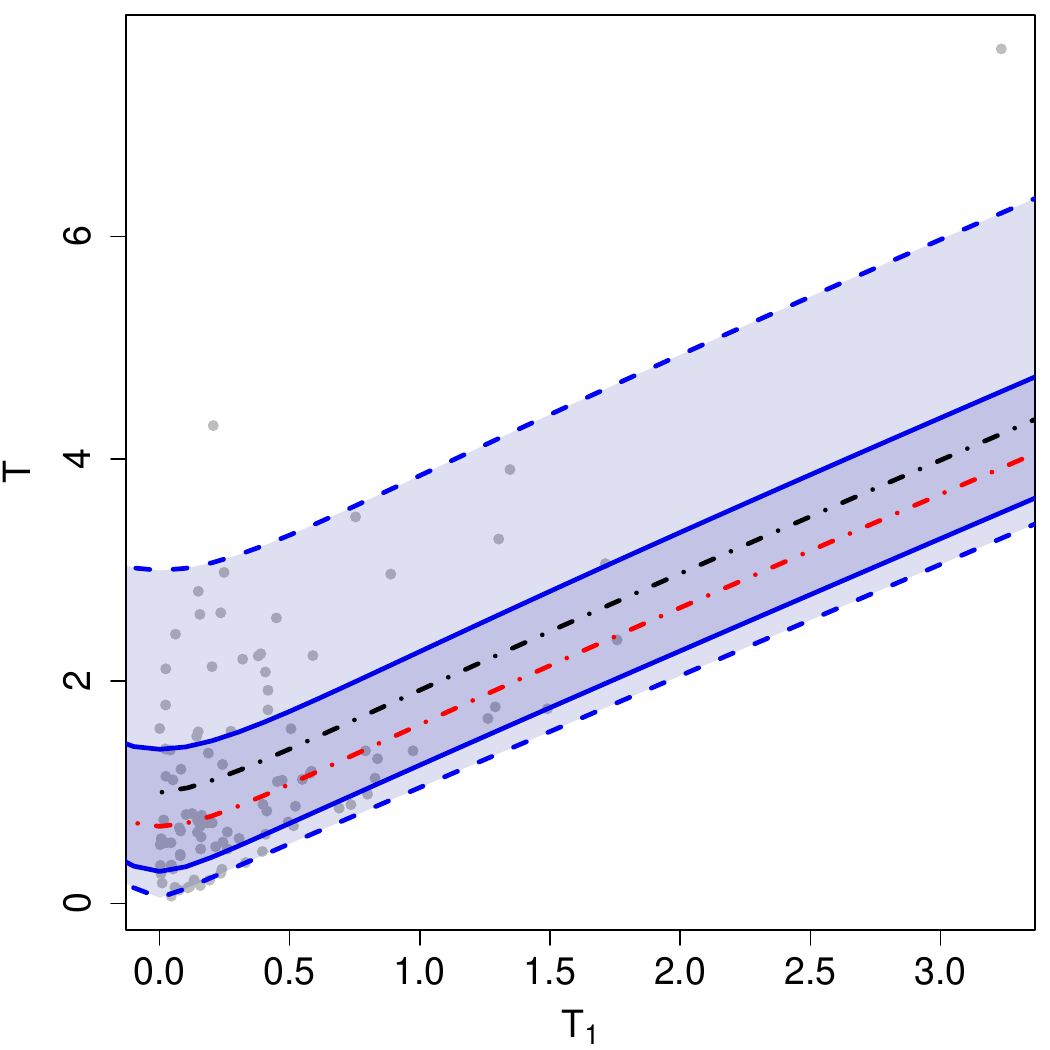}	\includegraphics*[scale=0.338]{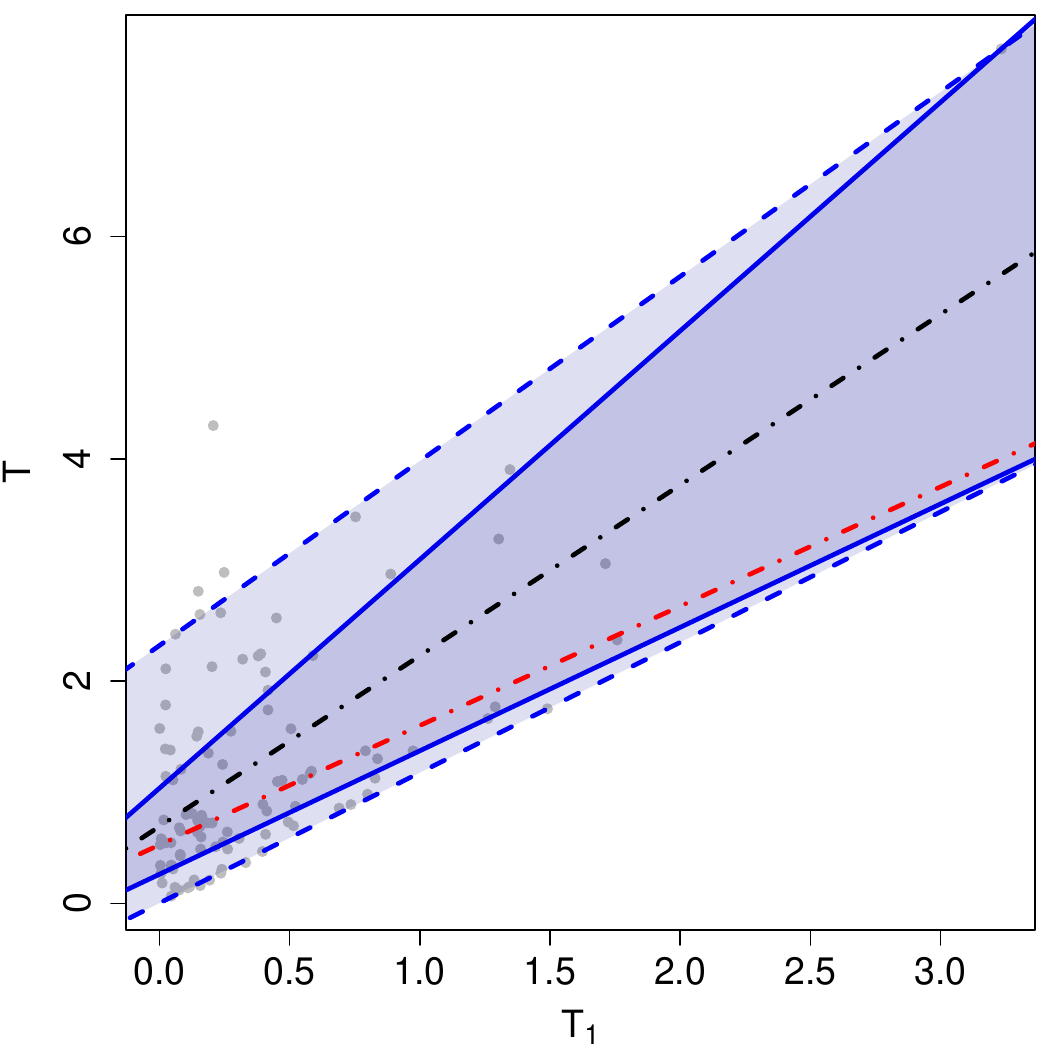}
		\caption{Scatterplots of a sample from $(T_1,T)$ for the systems in Example \ref{ex2} jointly with the theoretical (left) and estimated (right) median \textcolor{blue}{(red-dot-dashed line)} and mean \textcolor{blue}{(black-dot-dashed line)} regression curves and  centered prediction bands with confidence levels $50\%$ (dark grey) and $90\%$ (light grey).} \label{fig2}
	\end{center}
\end{figure}

In the following example, we show how to manage the different options in case II for IID components.

\begin{example}\label{ex3}
	Let us study the system   
	$T=\min(X_1,\max(X_2,X_3))$. The minimal path sets are $P_1=\{1,2\}$ and $P_2=\{1,3\}$. From \eqref{MPS}, we get
	\begin{equation}\label{eq1II}
	\bar F_T(t)=\Pr(X_{\{1,2\}}>t)+\Pr(X_{\{2,3\}}>t)-\Pr(X_{\{1,2,3\}}>t)
	\end{equation}
	for $t\geq 0$ ($1$ elsewhere). Now we assume that the  component lifetimes are IID and that $\bar F$ is their common reliability function. Then $\bar F_T(t)=\bar q(\bar F(t))$, where $\bar q(u)=2u^2-u^3$ for $u\in[0,1]$ and the expected lifetime of $T$ at time $t=0$ is 
	$$E(T)=\int_0^\infty \bar q(\bar F(t)) dt=2\int_0^\infty \bar F^2(t) dt-\int_0^\infty \bar F^3(t) dt.$$
	For example, if $\bar F(t)=\exp(-t/\mu)$ for $t\geq 0$ (exponential distribution with mean $\mu$), then $E(T)= 2\mu/3=0.666667\mu$. As in the preceding examples we choose $T_1=X_{1:3}$, that is, it is the first component failure. However, now $$\Pr(T=T_1)=\Pr(T_1=X_1)=1/3$$ and $(T_1,T)$ have a singular part at $T=T_1$ with probability $1/3$ (even when the component lifetimes are IID and absolutely continuous). Therefore, we are in case II.
	
	The joint reliability function of $(T_1,T)$ is
	$$\bar G(x,y)=\Pr(T_1>x,T>y)=\Pr(T_1>x)=\bar F^3(x)$$
	for $0\leq y\leq x$, and
	\begin{align*}
	\bar G&(x,y)
	=\Pr(X_{1:3}>x,X_{\{1,2\}}>y)+\Pr(X_{1:3}>x, X_{\{1,3\}}>y)-\Pr(X_{1:3}>y)\\
	&=\Pr(X_1>y,X_2>y,X_3>x)+\Pr(X_{1}>y, X_{2}>x,X_{3}>y)-\Pr(X_{1:3}>y)\\
	&=2\bar F(x)\bar F^2(y)-\bar F^3(y)
	\end{align*}
	for $0\leq x< y$. Note that $\bar G$ is continuous but not absolutely continuous. Moreover, it can be represented as  $\bar G(x,y)=\widehat D(\bar F(x),\bar F(y))$ for all $x,y$, where
	$$\widehat D(u,v)=\left\{\begin{array}{crr}
	u^3 & \text{for}&0\leq u \leq v\leq 1;\\
	2uv^2-v^3 & \text{for}&0\leq v <  u\leq 1;\\
	\end{array}%
	\right.$$
	and
	$$\partial_1 \widehat D(u,v)=\left\{\begin{array}{crr}
	3u^2 & \text{for}&0\leq u < v\leq 1;\\
	2v^2& \text{for}&0\leq v <  u\leq 1.\\
	\end{array}%
	\right.$$ 
 	
	To solve case $II.b$, we use \eqref{G1b} obtaining
$$
\bar G_{T|T_1}(y|t)=\Pr(T>y|T_1=t)=\frac {2\bar F^2(y)}{3\bar F^2(t)}
$$
for $y>t$ (one for $0\leq y\leq t$). Note that 
$$\alpha(t)=\Pr(T>T_1|T_1=t)=\lim_{y\to t^+}\bar G_{T|T_1}(y|t)=\frac 2 3,$$	
and that $1-\alpha(t) =\Pr(T=T_1|T_1=t)=1/3$. In this case, they do not depend on $t$ and so they coincide with $\Pr(T>T_1)$ and 	$\Pr(T=T_1)$, respectively. Then the median regression curve is
$$m(t)=\bar G^{-1}_{T|T_1}(0.5|t)=\bar F^{-1}\left( \sqrt{0.75} \bar F(t) \right)
$$
for $t\geq 0$. In the exponential case, we get
$$m(t)=t-0.5 \mu \ln(0.75)=t+0.143841\mu$$ for $t\geq 0$. However,  the mean regression curve (in the exponential case) is
$$\tilde m(t)=\int_0^\infty \bar G_{T|T_1}(y|t) dy=t+\int_0^\infty \frac {2\bar F^2(y)}{3\bar F^2(t)} dy=t+\frac 1 3 \mu 
$$
for $t\geq 0$. 
The prediction bands can be obtained in a similar way for any $\bar F$. For example, the $90\%$ bottom prediction band is 
$$I^{bottom}_{90}(t)=\left[t,\bar F^{-1}\left( \sqrt{0.15} \bar F(t) \right)\right]$$
for $t\geq 0$.  In the exponential case, it is
$$I^{bottom}_{90}(t)=\left[t,t-0.5 \mu \ln(0.15)\right]=\left[t,t+0.94856\mu \right]$$
for $t\geq 0$. Of course, the $50\%$ bottom prediction band is $I^{bottom}_{50}(t)=[t,m(t)]$.

In Figure \ref{fig3}, left,  we provide a scatterplot of a simulated sample of size $100$ from  $(T_1,T)$ jointly with the median \blue{(red-dot-dashed line)} and mean \blue{(black-dot-dashed line)} regression curves and the bottom $50\%$ (dark grey) and $90\%$ (light grey) prediction bands for this case. 
They contain $60$ and $87$ data from our simulated sample, respectively, including the $39$ data where $T_1=T$. Note again that the prediction bands give a better representation of the uncertainty in the system lifetime values  than the curves  $m$ and  $\tilde m$. As mentioned above, they can be estimated from the data by using linear quantile regression techniques (see \cite{K05}).

Let us study  now the case II.a, that is, let us assume that the first component failure happens at  a time $t$ ($T_1=t$) and that we know that the system is still alive ($T>t$). Then we want to predict $T$ under these assumptions. To this end, from \eqref{G1bII},  we need to solve  
\begin{equation}\label{ex3G}
\Pr(T>y|T_1=t,T>t)=\frac {\partial_1 \widehat D (\bar F(t),\bar F(y))}{\alpha(t)\partial_1 \widehat D (\bar F(t),1) }=\frac {\bar F^2(y)}{\bar F^2(t)}=w
\end{equation}
for $y>t$ and $0<w<1$. Then the median regression curve for this case is 
$$m(t)= \bar F^{-1}\left( \sqrt{0.5} \bar F(t) \right)
$$
for $t\geq 0$. In the exponential case, we get
$$m(t)=t-0.5 \mu \ln(0.5)=t+0.3465736\mu$$ for $t\geq 0$. However,  the mean regression curve in the exponential case is
$$\tilde m(t)=t+\int_0^\infty \frac {\bar F^2(y)}{\bar F^2(t)} dy=t+0.5 \mu   
$$
for $t\geq 0$. The bottom prediction bands are obtained similarly. We provide the plots in Figure \ref{fig3}, right. We use the same simulated sample but now the $39$ data satisfying $T=T_1$ should be avoided. In these cases, at time $t$, we do not need to predict $T$ since $T=T_1=t$. For the $61$ remaining points, we get $54$ in the bottom $90\%$ prediction band (i.e. a $54/61=88.52459\%$ of the data with $T>T_1$). Only $7$ data are not contained in this band. For the other band, we get $32$ out of $61$ data (i.e. a $32/61=52.45902\%$).

\blue{Finally, note that if we also have information about the broken component at time $t$, then these predictions will not change (due to the symmetry in the system structure and the IID assumption). Thus, if we know that $T_1=X_2$, then the resulting system is $T=\min(X_1,X_3)$, and therefore we get \eqref{ex3G} as well. The predictions for the other case  $T_1=X_3$ are the same. Note that $T_1$ cannot be equal to $X_1$ if we assume $T>T_1$. This is not always the case for other systems where the knowledge of the broken components can be used to get better preditions for the system failure. As mentioned before, these cases can be solved by using the techniques developed in \cite{NAS19,ND17} and the approach based on quantile regression presented here.}

\end{example}

\begin{figure}
	\begin{center}
		\includegraphics*[scale=0.338]{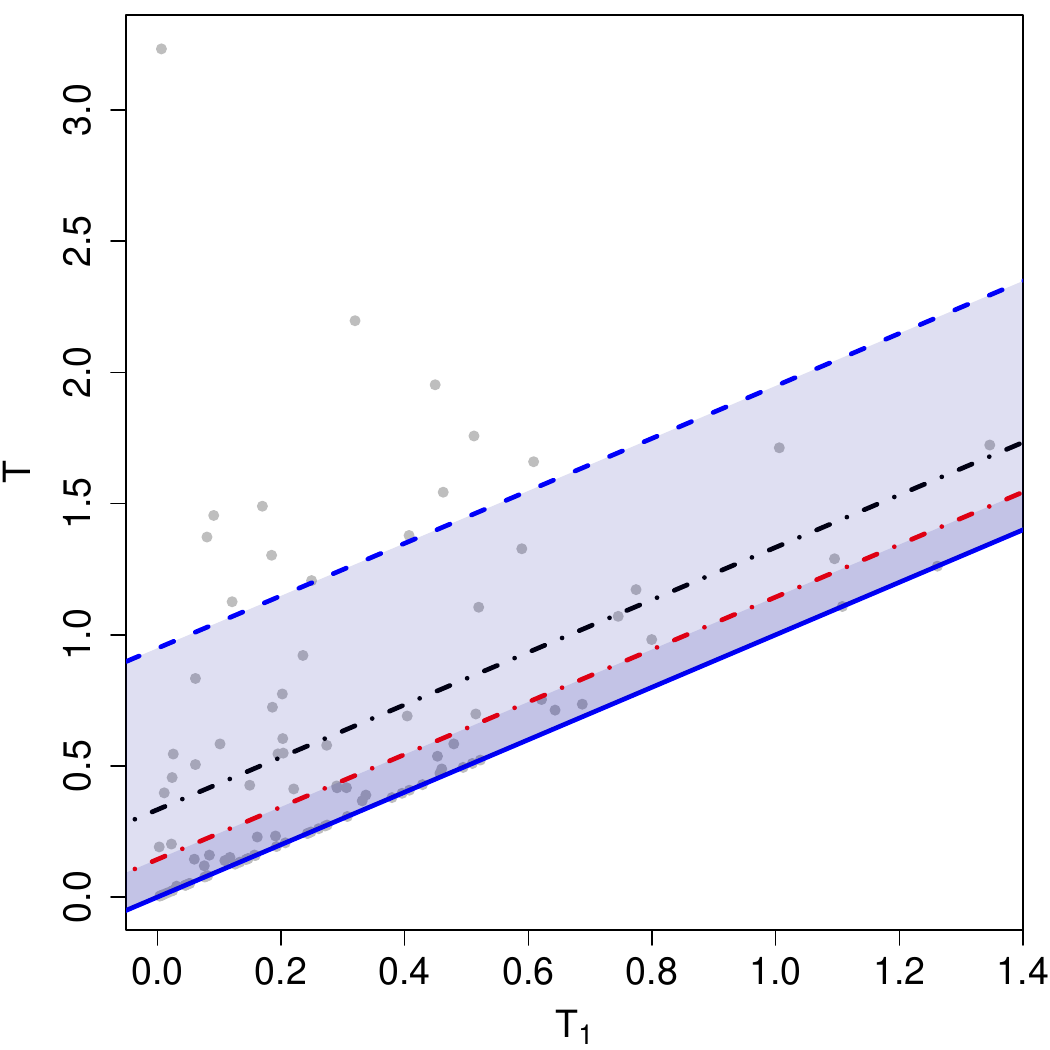}
        \includegraphics*[scale=0.338]{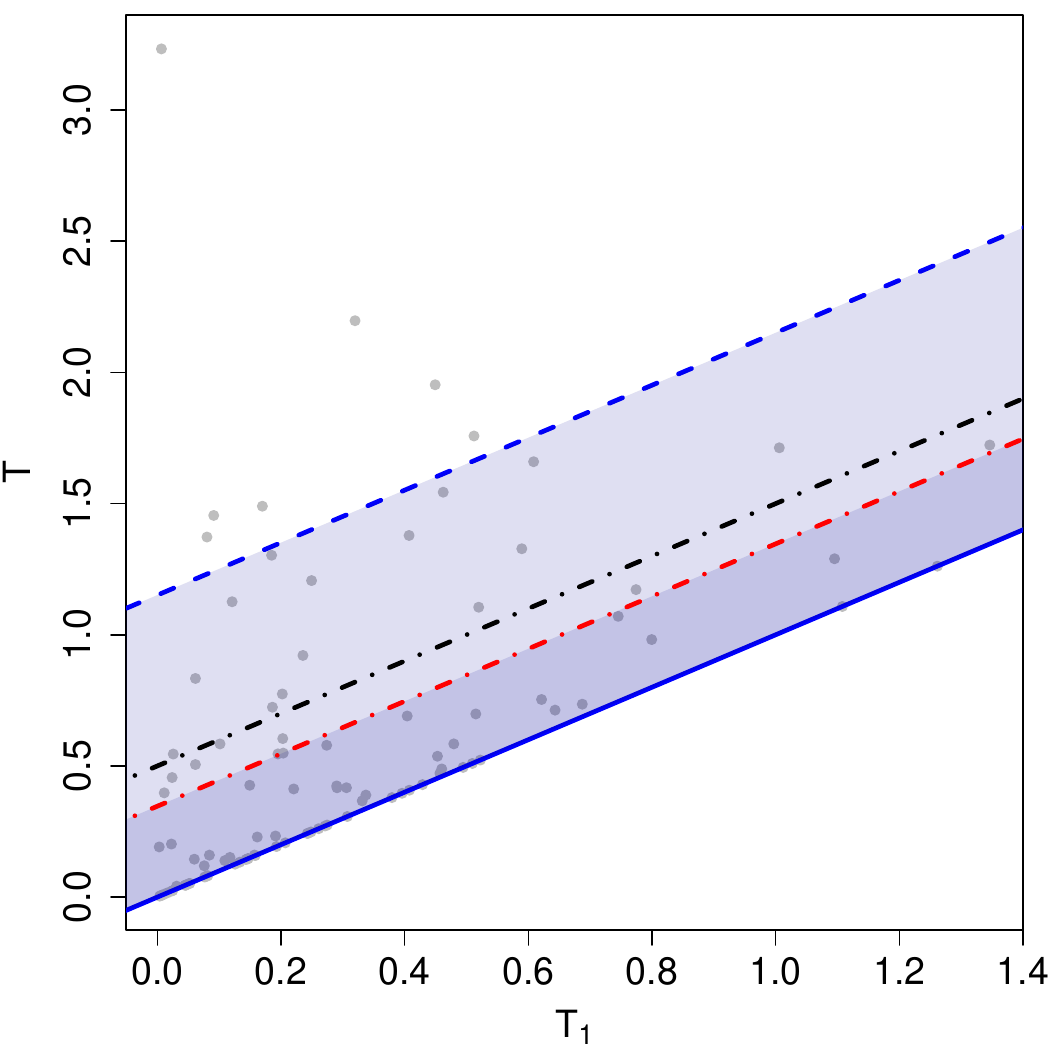}
		\caption{Scatterplots of a sample from $(T_1,T)$ for the systems in Example \ref{ex3} jointly with the plots of \textcolor{blue}{the} theoretical median \textcolor{blue}{(red-dot-dashed line)} and mean \textcolor{blue}{(black-dot-dashed line)} regression curves and the bottom prediction bands with confidence levels  $50\%$ (dark grey) and $90\%$ (light grey) for cases II.b (left) and II.a (right).} \label{fig3}
	\end{center}
\end{figure}

In the next example, we study case II \textcolor{blue}{in a system with} three dependent components. It also shows how to proceed when the explicit expression for the inverse function is not available.

\begin{example}\label{ex4}
	Let us consider the  same system  of Example \ref{ex3}, but now the component lifetimes  $(X_1,X_2,X_3)$ are ID and have the following Farlie-Gumbel-Morgenstern (FGM) survival copula
\begin{equation}\label{FGM}
	\widehat C(u_1,u_2,u_3)= u_1u_2u_3+\theta u_1u_2u_3(1-u_1)(1-u_2)(1-u_3)
\end{equation}
for $u_1,u_2,u_3\in[0,1]$ and $\theta\in [-1,1]$. When $\theta>0$ ($<$) the components have a positive (negative) correlation.
When $\theta=0$, the components are independent.
As in Example \ref{ex3}, we choose $T_1=X_{1:3}$ with $T\geq T_1$ (case II).  
Then, from \eqref{eq1II},  the respective reliability functions of $T$  and $T_1$ are
	\begin{equation*}
	\bar F_T(t)=2\bar F^2(t)-\bar F^3(t)-\theta \bar F^3(t) F^3(t)
	\end{equation*}
and
\begin{equation*}
\bar F_{T_1}(t)=\bar F^3(t)+\theta \bar F^3(t) F^3(t)
\end{equation*}
for $t\geq 0$ (one elsewhere). Therefore, $\bar F_T(t)=\bar q(\bar F(t))$, where $\bar q(u)=2u^2-u^3-\theta u^3(1-u)^3$ for $u\in[0,1]$. The expected lifetime of $T$ (at time $t=0$) when $\theta=1$ (positive dependence) and   $\bar F(t)=\exp(-t/\mu)$ for $t\geq 0$ (exponential distribution with mean $\mu$), is  $E(T)= 0.65\mu$. Note that we also have $\Pr(T=T_1)=\Pr(T_1=X_1)=1/3$ since $(X_1,X_2,X_3)$ has an exchangeable (symmetric) joint distribution.  

The joint reliability function of $(T_1,T)$ is
$$\bar G(x,y)=\Pr(T_1>x,T>y)=\Pr(T_1>x)=\bar F^3(x)+\theta \bar F^3(x) F^3(x)$$
for $0\leq y\leq x$ and
\begin{align*}
	\bar G(x,y)
	&=\Pr(X_{1:3}>x,X_{\{1,2\}}>y)+\Pr(X_{1:3}>x, X_{\{1,3\}}>y)-\Pr(X_{1:3}>y)\\
	&=2\Pr(X_1>y,X_2>y,X_3>x)-\Pr(X_{1:3}>y)\\
	&=2\bar F(x)\bar F^2(y)+2\theta \bar F(x)\bar F^2(y) F(x) F^2(y)
	-\bar F^3(y)-\theta \bar F^3(y) F^3(y)
	\end{align*}
	for $0\leq x< y$. Thence $\bar G$ is continuous but not absolutely continuous. Moreover, $\bar G(x,y)=\widehat D(\bar F(x),\bar F(y))$ for all $x,y$, where
	$$\widehat D(u,v)=\left\{\begin{array}{crr}
	u^3+\theta (u-u^2)^3 & \text{for}&0\leq u \leq v\leq 1;\\
	2uv^2+2\theta (u-u^2)(v-v^2)^2-v^3-\theta (v-v^2)^3 & \text{for}&0\leq v <  u\leq 1;\\
	\end{array}%
	\right.$$
	and
	$$\partial_1 \widehat D(u,v)=\left\{\begin{array}{crr}
	3u^2 +3\theta u^2(1-u)^2(1-2u)& \text{for}&0\leq u < v\leq 1;\\
	2v^2+2\theta v^2(1-v)^2 (1-2u)& \text{for}&0\leq v <  u\leq 1.\\
	\end{array}%
	\right.$$ 
	
	To solve case $II.b$, we use \eqref{G1b} obtaining
\begin{equation}\label{ex4eq1}
	\bar G_{T|T_1}(y|t)=\Pr(T>y|T_1=t)=\frac {2}{3}
\cdot \frac{\bar F^2(y)+\theta \bar F^2(y)  F^2(y) (1-2F(t))}
{\bar F^2(t) +\theta \bar F^2(t) F^2(t)(1-2F(t))}
\end{equation}
	for $y>t$ (one for $0\leq y\leq t$). Note that 
	$$\alpha(t)=\Pr(T>T_1|T_1=t)=\lim_{y\to t^+}\bar G_{T|T_1}(y|t)=\frac 2 3$$	
	and that $1-\alpha(t) =\Pr(T=T_1|T_1=t)=1/3$. In this case, they do not depend on $t$ as well (due to the symmetry of the model and the system). 	To get the inverse $\bar G^{-1}_{T|T_1}(w|t)$ of this function we need to solve in $y$ the equation 
	$$\frac{\bar F^2(y)+\theta \bar F^2(y)  F^2(y) (1-2F(t))}
	{\bar F^2(t) +\theta \bar F^2(t) F^2(t)(1-2F(t))}	=\frac 32 w$$
for $0<w<1$. Unfortunately, we do not have an explicit expression for this solution. Instead, we can use \eqref{ex4eq1} to get the plots of the level curves of $\bar G_{T|T_1}(y|t)$. Thus, the median regression curve $m$ is obtained with the level $w=0.5$ and the $90\%$ bottom prediction band with the level $w=0.1$. 
In Figure \ref{fig4}, left,  we plot these level curves  and  the associated prediction bands for $\theta=1$ and $\mu=1$. We also add a scatterplot of a simulated sample of size $100$ from $(T_1,T)$ obtained by using the inverse transform method (see, e.g., \cite{MS02}, p.\ 88).


To generate this sample we note that the survival copula of $(X_1,X_2)$ is 
$$\widehat C_{1,2}(u,v)=\widehat C(u,v,1)=uv$$
for $u,v\in[0,1]$.  Therefore, the values for $X_1$ and $X_2$ can be generated as $X_i=\bar F^{-1}(U_i)=-\ln(U_i)$,  where $U_i$ for $t=1,2$  are  
independent uniform random variables in $(0,1)$.
To obtain $X_3$ we note that the PDF of $(X_1,X_2,X_3)$ is 
$$f_{1,2,3}(x_1,x_2,x_3)=f(x_1)f(x_2)f(x_3) \partial_{1,2,3} \widehat C(\bar F_1(x_1),\bar F_2(x_2),\bar F_3(x_3)),$$
where
$$\partial_{1,2,3} \widehat C(u_1,u_2,u_3)=1+(1-2u_1)(1-2u_2)(1-2u_3)$$
for $u_1,u_2,u_3\in[0,1]$. Therefore, a PDF of $(X_3|X_1=x_1,X_2=x_2)$ is
$$f_{3|1,2}(x_3|x_1,x_2)=\frac{f_{1,2,3}(x_1,x_2,x_3)}{f(x_1)f(x_2)}=e^{-x_3}+(1-2e^{-x_1})(1-2e^{-x_2})(e^{-x_3}-2e^{-2x_3}) $$
for $x_3\geq 0$, where $x_1,x_2>0$. Hence,
$$\bar F_{3|1,2}(x_3|x_1,x_2)=e^{-x_3}+(1-2e^{-x_1})(1-2e^{-x_2})(e^{-x_3}-e^{-2x_3})$$
for $x_3\geq 0$. To get the inverse of this function we must solve (in $x_3$) the equation  $\bar G_{3|1,2}(x_3|x_1,x_2)=w$ for $0<w<1$, which is equivalent to
$$a(x_1,x_2) e^{-2x_3}- (1+a(x_1,x_2)) e^{-x_3} +w=0,$$
where $a(x_1,x_2)=(1-2e^{-x_1})(1-2e^{-x_2})$. Thus, we must solve $p(z)=0$ for the polynomial $p(z)=az^2-(a+1)z+w$. As $p(0)=w>0$ and $p(1)=a-(a+1)+w=w-1<0$, this polynomial has a unique solution in $(0,1)$.  Therefore, if $a(x_1,x_2)\neq 0$, then this root is
$$ z=e^{-x_3}=\frac{1+a(x_1,x_2)-\sqrt{(1+a(x_1,x_2))^2-4wa(x_1,x_2)}}{2a(x_1,x_2)},$$
that is,
\begin{equation}\label{INVF3|21}
\bar F^{-1}_{3|1,2}(w|x_1,x_2)=-\ln\left( \frac{1+a(x_1,x_2)-\sqrt{(1+a(x_1,x_2))^2-4wa(x_1,x_2)}}{2a(x_1,x_2)} \right)
\end{equation}
for $0<w<1$, $x_1,x_2>0$ and $a(x_1,x_2)\neq 0$. If $a(x_1,x_2)= 0$, then we obtain $\bar F^{-1}_{3|1,2}(w|x_1,x_2)=-\ln(w)$. Note that the event $a(X_1,X_2)= 0$ has probability zero when $X_1$ and $X_2$ are generated randomly as assumed above. Therefore, $X_3$ can be generated for given values of $X_1$ and $X_2$ from \eqref{INVF3|21}  as 
$X_3:= \bar F^{-1}_{3|1,2}(U_3|X_1,X_2)$ where $U_3$  in an independent random number in $(0,1)$.


\quad

In our simulated sample, we get $31$ points where $T_1=T$ and the sample mean $0.6930643$ for $T$. These values are close to the expected values $33.333$ and $0.65$, respectively.

\quad

Now we study the case II.a, that is, we assume that the first component failure happens at a time $t$ ($T_1=t$), and at this time we know that the system is still alive  ($T>t$). Then we want to predict $T$ under these assumptions. To this end, from \eqref{G1bII},  we need to solve

	\begin{equation*}\label{ex4G}
	\Pr(T>y|T_1=t,T>t)= \frac{\bar F^2(y)+\theta \bar F^2(y)  F^2(y) (1-2F(t))}
	{\bar F^2(t) +\theta \bar F^2(t) F^2(t)(1-2F(t))}=w
	\end{equation*}
	for $y>t$ and $0<w<1$. 
Again, we do not have an explicit solution. Therefore we can use contour plots. Thus, the median regression curve is obtained with the level curve with $w=1/2$.  The bottom prediction bands are obtained in a similar way. They are plotted in Figure \ref{fig4}, right. We use the same simulated sample but now the $31$ data satisfying $T=T_1$ should be avoided. Only two data points are out of our $90\%$ bottom prediction band.

\end{example}

\begin{figure}
	\begin{center}
\includegraphics*[scale=0.338]{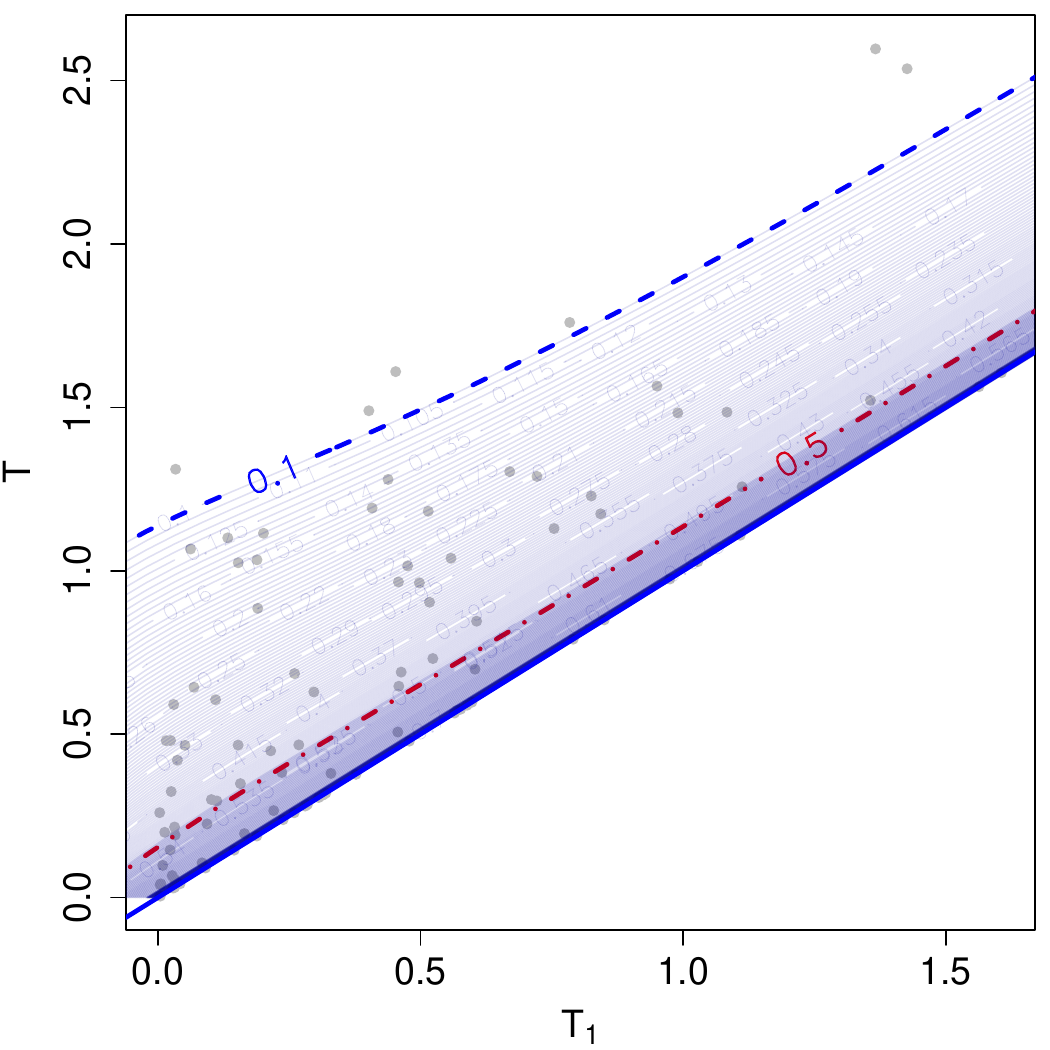}	
\includegraphics*[scale=0.338]{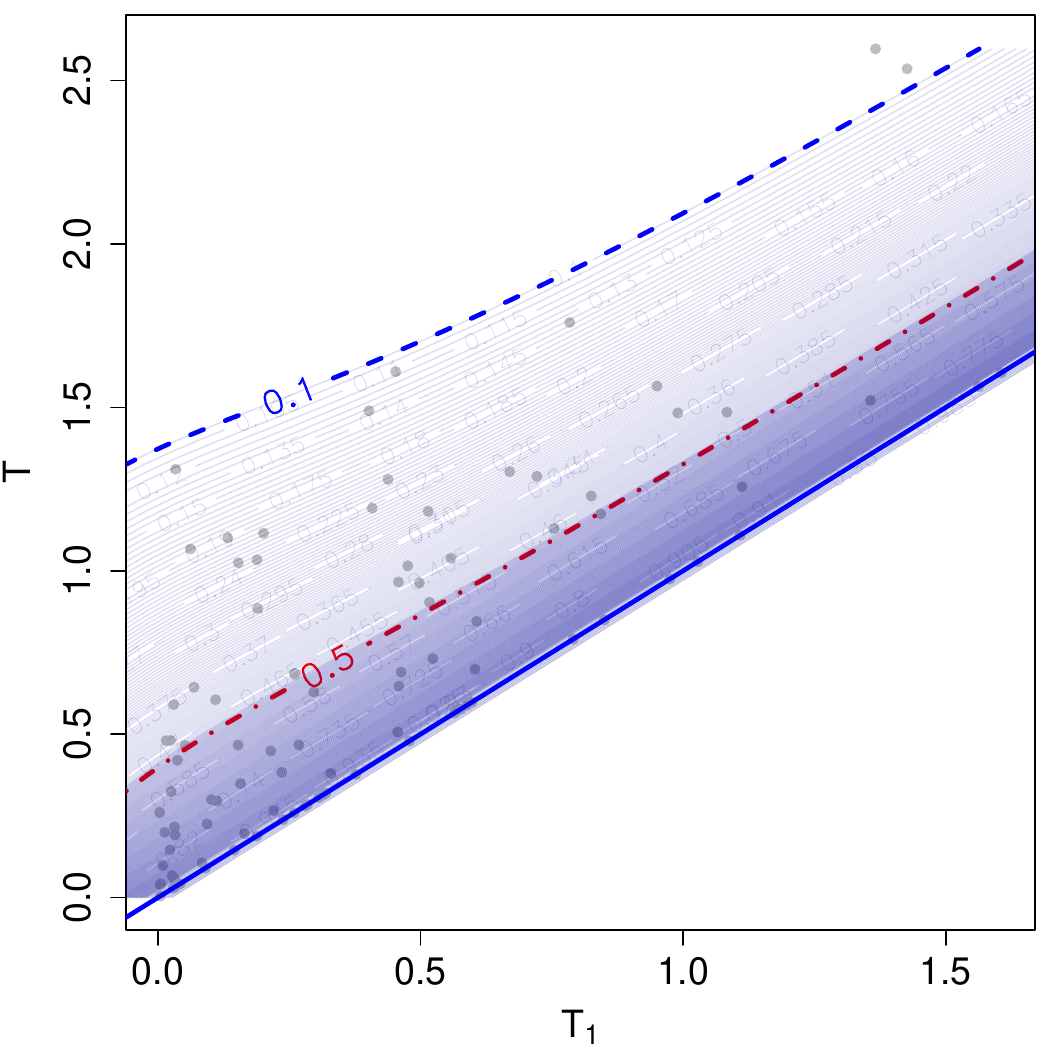}
		\caption{Scatterplots of a simulated sample from $(T_1,T)$ for the systems in Example \ref{ex4} jointly with the theoretical median regression curve \textcolor{blue}{(red-dot-dashed line)} and $50\%$ (dark grey) and $90\%$ (light grey) bottom prediction bands for cases II.b (left) and II.a (right).} \label{fig4}
	\end{center}
\end{figure}

The following examples show how to proceed in case III. The first one is based on the order statistics from IID samples of size $3$. These results are well known (see \cite{DN03,NB22}). Similar results for record values can be seen in \cite{N21b}. 

\begin{example}\label{ex5}%
Let us study the parallel system $T=\max(X_1,X_2,X_3)=X_{3:3}$ with IID components and common reliability function $\bar F$. Then, we know that the first component failure occurs \blue{at time} $t_1>0$ and that the second failure happens \blue{at time} $t_2>t_1$. So we are in case III with $T_1=X_{1:3}=t_1$ and $T_2=X_{2:3}=t_2$. Note that $T>T_2$ (with probability one). 

As $(T_1,T_2,T)$ are the order statistics from three IID random variables, then it is well known (see, e.g., \cite{DN03}, p.\ 12) that their joint PDF is
$$g(t_1,t_2,t)=6 f(t_1)f(t_2)f(t)$$
for $0\leq t_1\leq t_2\leq t$. 
The joint reliability of $(T_1,T_2)$ is 
$$\bar G_{1,2}(t_1,t_2)=\Pr(X_{1:3}>t_1,X_{2:3}>t_2)=3\bar F(t_1)\bar F^2(t_2)-2\bar F^3(t_2)$$ 
for $0\leq t_1\leq t_2$. Therefore, their joint PDF is 
$$g_{1,2}(t_1,t_2)=\partial_{1,2} \bar G_{1,2}(t_1,t_2)=6f(t_1)f(t_2)\bar F(t_2)$$ 
for $0\leq t_1\leq t_2$ (zero elsewhere). Hence, 
the conditional PDF of $(T|T_1=t_1,T_2=t_2)$ for $0\leq t_1\leq t_2$ is
$$g_{3|1,2}(t|t_1,t_2)=\frac{g(t_1,t_2,t)}{g_{1,2}(t_1,t_2)}=\frac{f(t)}{\bar F(t_2)}$$
for $t\geq t_2$ (zero elsewhere). Note that it does not depend on $t_1$. This property is the well-known Markovian property of the order statistics (see, e.g., \cite{DN03}, p.\ 17). Hence, the conditional reliability function is 
$$\bar G_{3|1,2}(t|t_1,t_2)=\Pr(T>t|T_1=t_1,T_2=t_2)=\frac{\bar F(t)}{\bar F(t_2)}$$
for $t\geq t_2$ (one elsewhere) and its quantile function is
$$\bar G^{-1}_{3|1,2}(w|t_1,t_2)=\bar F^{-1}(w\bar F(t_2))$$
for $0<w<1$. As in the preceding examples, this expression can be used to get the median regression curve and the desired prediction bands. The reliability function $\bar G_{3|1,2}$ can also be obtained from distortions by using \eqref{GcaseIII}.

\blue{Similar results hold for general $k$-out-of-$n$ systems (i.e. systems that fail when at least $k$ of its $n$ components fail) with IID components. The results were obtained in \cite{BR22,NB22}. Thus, if we want to predict $X_{s:n}$ from $X_{r:n}$ for $1\leq r< s\leq n$, we can use that 
$$\Pr(X_{s:n}>y|X_{r:n}=t)=\Pr(X_{s-r:n-r}>y|X_{1:n-r}>t)=\bar q_{s-r:n-s}\left(\frac{\bar F(y)}{\bar F(x)}\right)$$
for $y\geq t\geq 0$, where $\bar q_{s-r:n-s}$ is the distortion function of $X_{s-r:n-r}$ under the IID assumption (see, e.g., \cite{N21}, p.\ 30). This function is a polynomial, therefore, it is not possible, in general, to obtain its inverse to get explicit expressions for the median regression curves and the prediction bands. Of course, we could use here numerical solutions instead. Alternatively, we can use Propositions 2.1 and 2.2 in \cite{NB22}. Thus, the median regression curve to predict $X_{s:n}$ from $X_{r:n}=t$ can be obtained as
$$m(t)=\bar F^{-1}(\beta_{0.5} \bar F(t)), $$
where $\beta_{0.5}$ is the median of a beta distribution with parameters $n-s+1$ and $s-r$ (that distribution is included in many statistical programs like R).  For example, if our system is $T=X_{5:10}$ and it is equipped with a warning alarm at the second component failure $T_1=X_{2:10}$, then the median regression curve to predict $X_{5:10}$ from $X_{2:10}=t$ is 
$$m(t)=\bar F^{-1}(0.679481\bar F(t)).$$ 
The prediction bands can be obtained similarly from expression (2.2) in \cite{NB22}.}
\end{example}

In the following example, we study how a weak dependence affects the predictions for a parallel system with three ID components in case III.

\begin{example}\label{ex6}%
We study the same system structures as in Example \ref{ex5}, that is,  $T=X_{3:3}$,  $T_1=X_{1:3}=t_1$ and $T_2=X_{2:3}=t_2$ for $0<t_1<t_2$. However, now we consider that $(X_1,X_2,X_3)$ can be dependent. To simplify the expressions let us assume that they are exchangeable (i.e. their joint distribution is invariant in law under permutations). Then, all components are ID with common reliability function $\bar F$ and their survival copula $\widehat C$ is also exchangeable (symmetric under permutations of their variables). Under these assumptions, a straightforward calculation shows that the joint reliability function $\bar G$ of $(T_1,T_2,T)$ can be written as
$$\bar G(t_1,t_2,t)=6\mathbf{\bar F}(t_1,t_2,t)-3\mathbf{\bar F}(t_2,t_2,t)-3\mathbf{\bar F}(t_1,t,t)+\mathbf{\bar F}(t,t,t)$$  for $0\leq t_1\leq t_2\leq t$, where $\mathbf{\bar F}(x_1,x_2,x_3)=\Pr(X_1>x_1,X_2>x_2,X_3>x_3)$ is the joint reliability function of $(X_1,X_2,X_3)$.	Therefore, $\bar G$ can be represented as $\bar G(t_1,t_2,t)=\widehat D(\bar F(t_1),\bar F(t_2),\bar F(t))$ for $0\leq t_1\leq t_2\leq t$, where
$$\widehat D(u,v,w)= 6\widehat C(u,v,w)-3\widehat C(v,v,w)-3\widehat C(u,w,w)+\widehat C(w,w,w)$$ 
for $0\leq w\leq v\leq u\leq 1$. The expressions for $\widehat D$ in the other cases can be obtained similarly.

Analogously,  the joint reliability function $\bar G_{1,2}$ of $(T_1,T_2)$ is
$$\bar G_{1,2}(t_1,t_2)=3\mathbf{\bar F}(t_1,t_2,t_2)-2\mathbf{\bar F}(t_2,t_2,t_2)$$
for $0\leq t_1\leq t_2$, that is, $\bar G_{1,2}(t_1,t_2)=\widehat D(\bar F(t_1),\bar F(t_2),1)$ with
$$\widehat D(u,v,1)= 3\widehat C(u,v,v)-2\widehat C(v,v,v)$$
for $0\leq v\leq u\leq 1$. 

Therefore, by differentiating these expressions we get
$$\partial_{1,2} \widehat D(u,v,w)=6\partial_{1,2}\widehat C(u,v,w)$$
and
$$\partial_{1,2} \widehat D(u,v,1)=6\partial_{1,2}\widehat C(u,v,v)$$
and, by using \eqref{GcaseIII}, the  reliability function of $(T|T_1=t_1,T_2=t_2)$ is
$$\bar G_{3|1,2}(t|t_1,t_2)=\frac{\partial_{1,2}\widehat C(\bar F(t_1),\bar F(t_2),\bar F(t))}{\partial_{1,2}\widehat C(\bar F(t_1),\bar F(t_2),\bar F(t_2))}$$ 
for $t\geq t_2$ (one for $0\leq t< t_2$).

Note that if $\widehat C$ is the product copula (independent components), this expression leads to the expression obtained in Example \ref{ex5} because  
 $\partial_{1,2}\widehat C(u,v,w)=w$ holds for all $u,v,w\in(0,1)$.

Now, let us consider a different survival copula. For example, let us assume the FGM copula of Example \ref{ex4}, which represents a weak dependence between the component lifetimes. Then 
$$\partial_{1,2}\widehat C(u_1,u_2,u_3)=u_3+\theta u_3(1-u_3)(1-2u_1)(1-2u_2)$$
for all $u_1,u_2,u_3\in[0,1]$, and we get 
\begin{align*}
\bar G_{3|1,2}(t|t_1,t_2)
&=\frac{\bar F(t)+\theta \bar F(t) F(t)(1-2\bar F(t_1))(1-2\bar F(t_2)) }{\bar F(t_2)+\theta \bar F(t_2)F(t_2) (1-2\bar F(t_1))(1-2\bar F(t_2)) }\\
&=\frac{\bar F(t)}{\bar F(t_2)}\cdot \frac{1+\theta  F(t)(1-2\bar F(t_1))(1-2\bar F(t_2)) }{1+\theta F(t_2) (1-2\bar F(t_1))(1-2\bar F(t_2)) }
\end{align*} 
for $t\geq t_2$ (one for $0\leq t< t_2$). Hence, for $\theta=0$, it  coincides with the expression for the IID case (as mentioned above). However, for $\theta\neq 0$, it depends on $t_1$  and so the Markovian property does not hold for this copula.
To get its inverse function, we need to solve the following quadratic equation
$$\theta A(t_1,t_2)  \bar F^2(t)-(1+\theta A(t_1,t_2))\bar F(t)+wB(t_1,t_2)=0,$$
where
$$A(t_1,t_2)=(1-2\bar F(t_1))(1-2\bar F(t_2))\in[-1,1]$$
and
$$B(t_1,t_2)=\bar F(t_2)+\theta\bar F(t_2) F(t_2)(1-2\bar F(t_1))(1-2\bar F(t_2))\in[0,1]$$
for all $t_1,t_2$. A straightforward calculation shows that this equation has a unique solution in $[0,1]$ given by 
$$\bar F(t)= \frac{1+\theta A(t_1,t_2)-\sqrt{(1+\theta A(t_1,t_2))^2-4\theta wA(t_1,t_2)B(t_1,t_2)}}{2\theta A(t_1,t_2)}$$
for $\theta A(t_1,t_2)\neq 0$. From this expression we can compute $\bar G^{-1}_{3|1,2}(w|t_1,t_2)$
for $0<w<1$, $0\leq t_1\leq t_2$ and  $\theta\in[-1,1]$. In particular, the median regression map is obtained with $w=0.5$.

For example, let us consider a standard exponential reliability function and $\theta=1$. Then
$$\bar F_T(t)= 3e^{-t}-3e^{-2t}+e^{-3t}+ e^{-3t}(1-e^{-t})^3$$
for $t\geq 0$ and its mean $E(T)=1.85$. 

In Figure \ref{fig5}, left,  we plot the level curves (predictions) of the median regression map $m(t_1,t_2)=\bar G^{-1}_{3|1,2}(0.5|t_1,t_2)$  jointly with  the values obtained for $T_1$ and $T_2$ in the simulated sample for $(X_1,X_2,X_3)$ of Example \ref{ex4} with the FGM copula and $\theta=1$.

\begin{figure}
	\begin{center}
	\includegraphics*[scale=0.338]{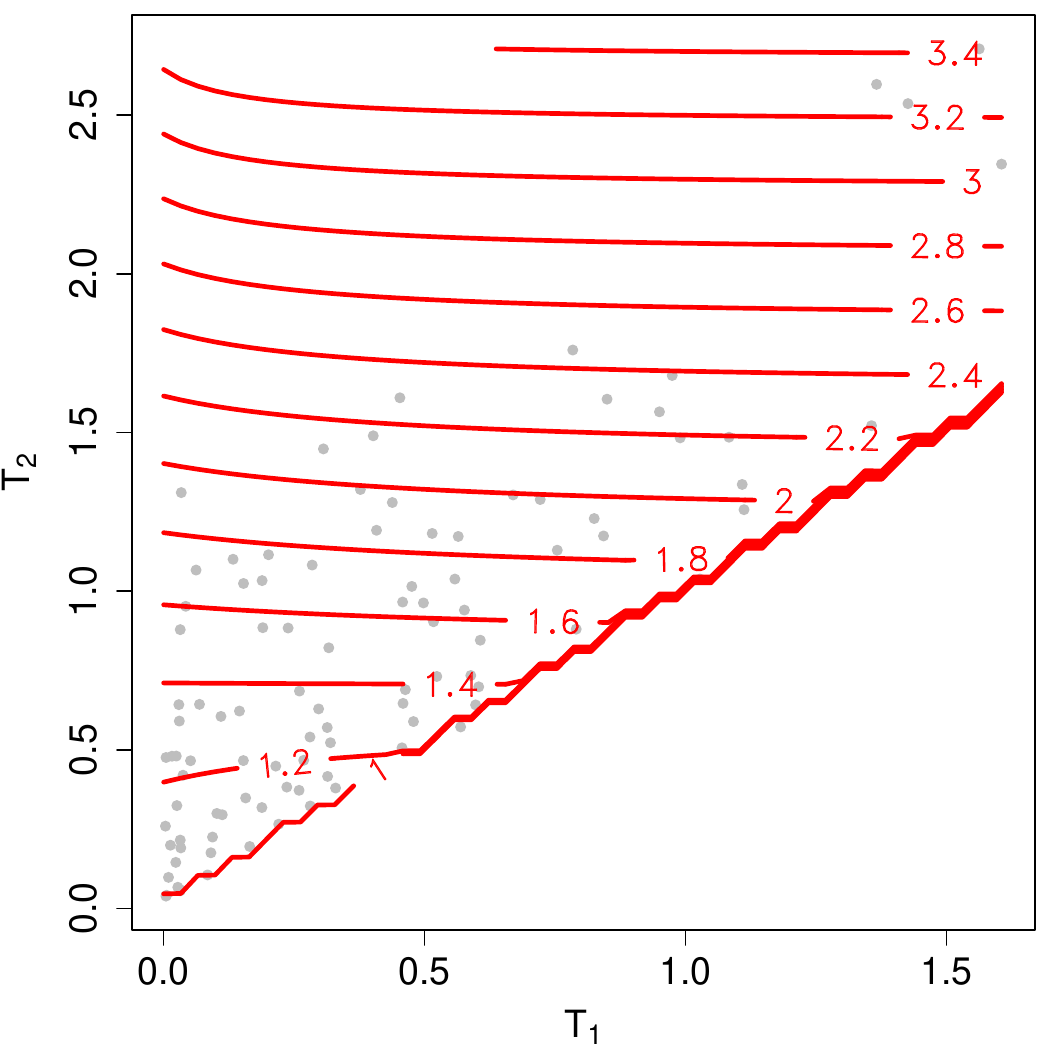}	
    \includegraphics*[scale=0.338]{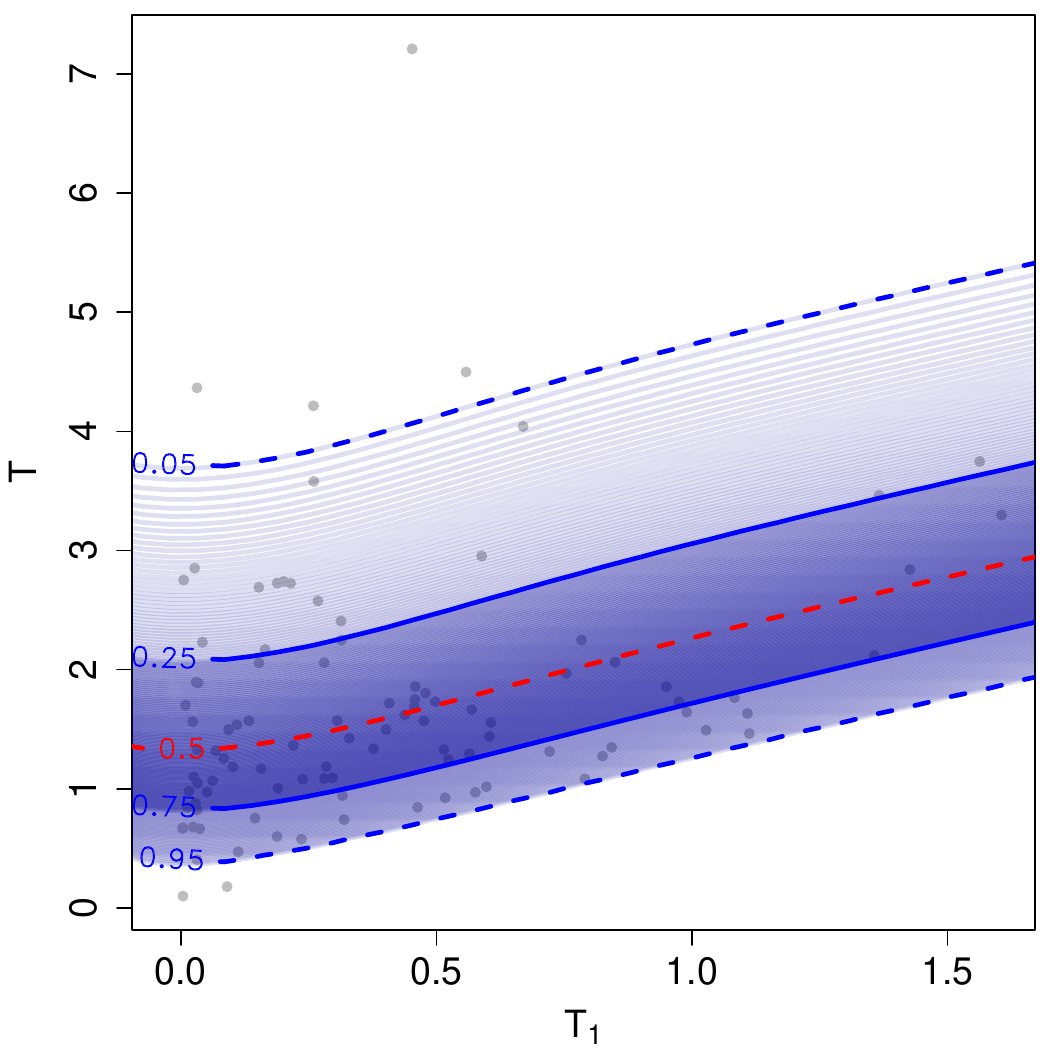}
		\caption{Scatterplots of a sample from $(T_1,T_2)$ (left) and  $(T_1,T)$ (right)  for the systems in Example \ref{ex6} jointly with the theoretical median  regression curve \textcolor{blue}{(red-dashed line)}   and the  centered prediction bands (right plot) with confidence le
  vels $50\%$ (dark grey) and $90\%$ (light grey). In the left plot, we only give the level curves (predictions) of the median regression map $m(t_1,t_2)$. } \label{fig5}
	\end{center}
\end{figure}

These predictions can be compared with the ones obtained from the first component failure, that is, $(T|T_1=t_1)$ (case I) given  Figure \ref{fig5}, right.
To get these predictions we observe that
$$\Pr(T_1>t_1,T>t)=\mathbf{\bar F}(t_1,t_1,t_1)$$
for $0\leq t< t_1$ and
$$\Pr(T_1>t_1,T>t)=3\mathbf{\bar F}(t_1,t_1,t)-3\mathbf{\bar F}(t_1,t,t)+\mathbf{\bar F}(t,t,t)$$
for $0\leq t_1\leq  t$. Therefore,  $\Pr(T_1>t_1,T>t)=\widehat D_{1,3}(\bar F(t_1), \bar F(t))$ where
$$\widehat D_{1,3}(u,w)=\left\{\begin{array}{crr}
\widehat C(u,u,u) & \text{for}&0\leq u \leq w\leq 1;\\
3\widehat C(u,u,w)-3\widehat C(u,w,w)+\widehat C(w,w,w) & \text{for}&0\leq w <  u\leq 1.\\
\end{array}%
\right.$$
For the FGM copula in \eqref{FGM} we get
$$\widehat D_{1,3}(u,w)=\left\{\begin{array}{crr}
u^3+\theta u^3 (1-u)^3 & \text{for}&0\leq u \leq w\leq 1;\\
3u^2w+3\theta u^2w(1-u)^2(1-w) & &\\
-3uw^2-3\theta uw^2(1-u)(1-w)^2& &\\ +w^3+\theta w^3 (1-w)^3 &\text{for}&0\leq w <  u\leq 1;\\
\end{array}%
\right.$$
and
$$\partial_1 \widehat D_{1,3}(u,v)=\left\{\begin{array}{crr}
3u^2 +3\theta u^2(1-u)^2(1-2u)& \text{for}&0\leq u < w\leq 1;\\
6uw+6\theta u(1-u)(1-2u)w(1-w)&&\\
 -3w^2-3\theta w^2(1-w)^2 (1-2u)& \text{for}&0\leq w <  u\leq 1.\\
\end{array}%
\right.$$
Hence, from \eqref{G1b}, we obtain
\begin{align*}
\bar G_{3|1}(t|t_1)&=\frac{6\bar F(t_1)\bar F(t)+6\theta \bar F(t_1)\bar F(t)F(t_1)F(t)(1-2\bar F(t_1))}{3\bar F^2(t_1)+3\theta \bar F^2(t_1)F^2(t_1)(1-2\bar F(t_1)) }\\
&\quad 
-\frac{3\bar F^2(t)+3\theta \bar F^2(t)F^2(t)(1-2\bar F(t_1))}{3\bar F^2(t_1)+3\theta \bar F^2(t_1)F^2(t_1)(1-2\bar F(t_1)) }
\end{align*}
for $t\geq t_1$ (one elsewhere). As in Example \ref{ex4}, we do not have an explicit expression for its inverse. Therefore, Figure \ref{fig5}, right, is plotted by using level curves (contour plots).

\begin{figure}
	\begin{center}
		\includegraphics*[scale=0.338]{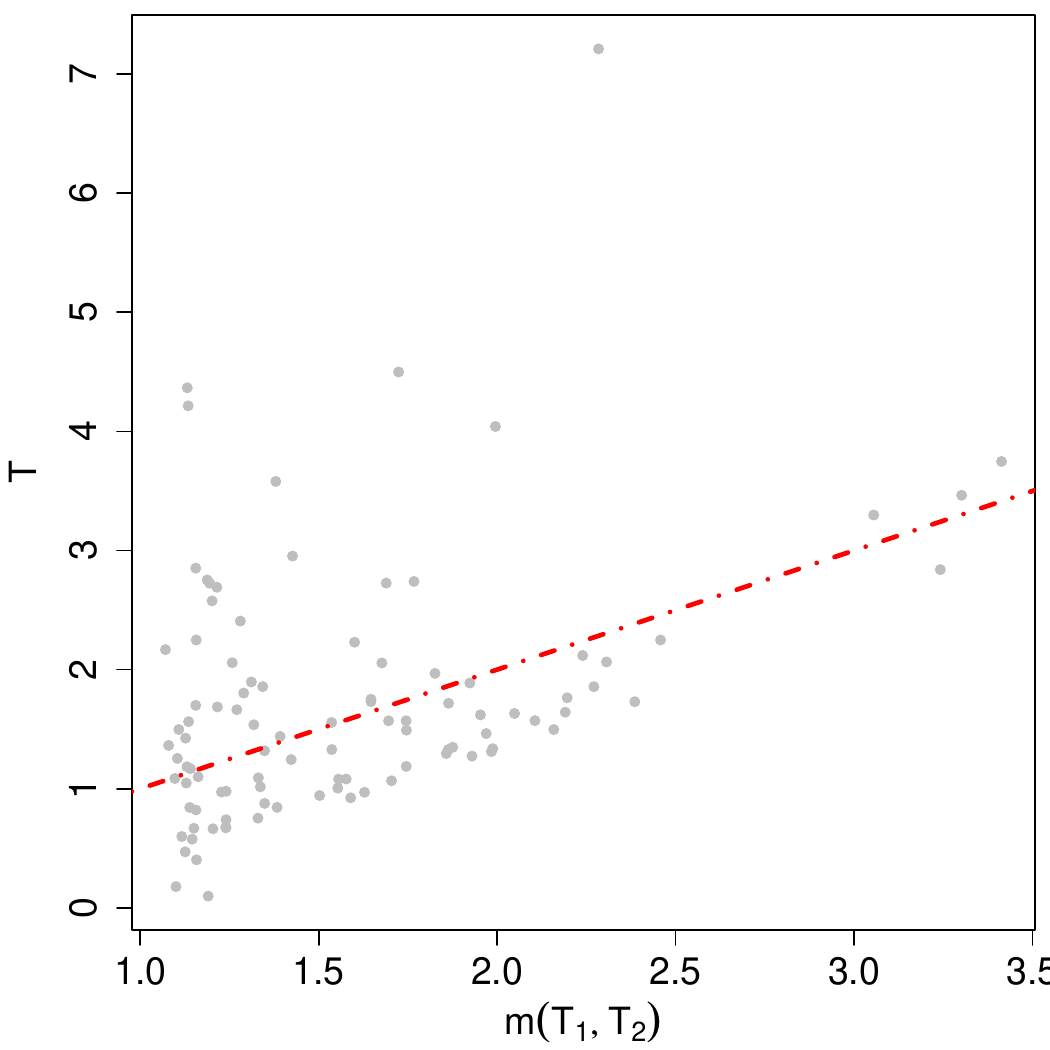}
        \includegraphics*[scale=0.338]{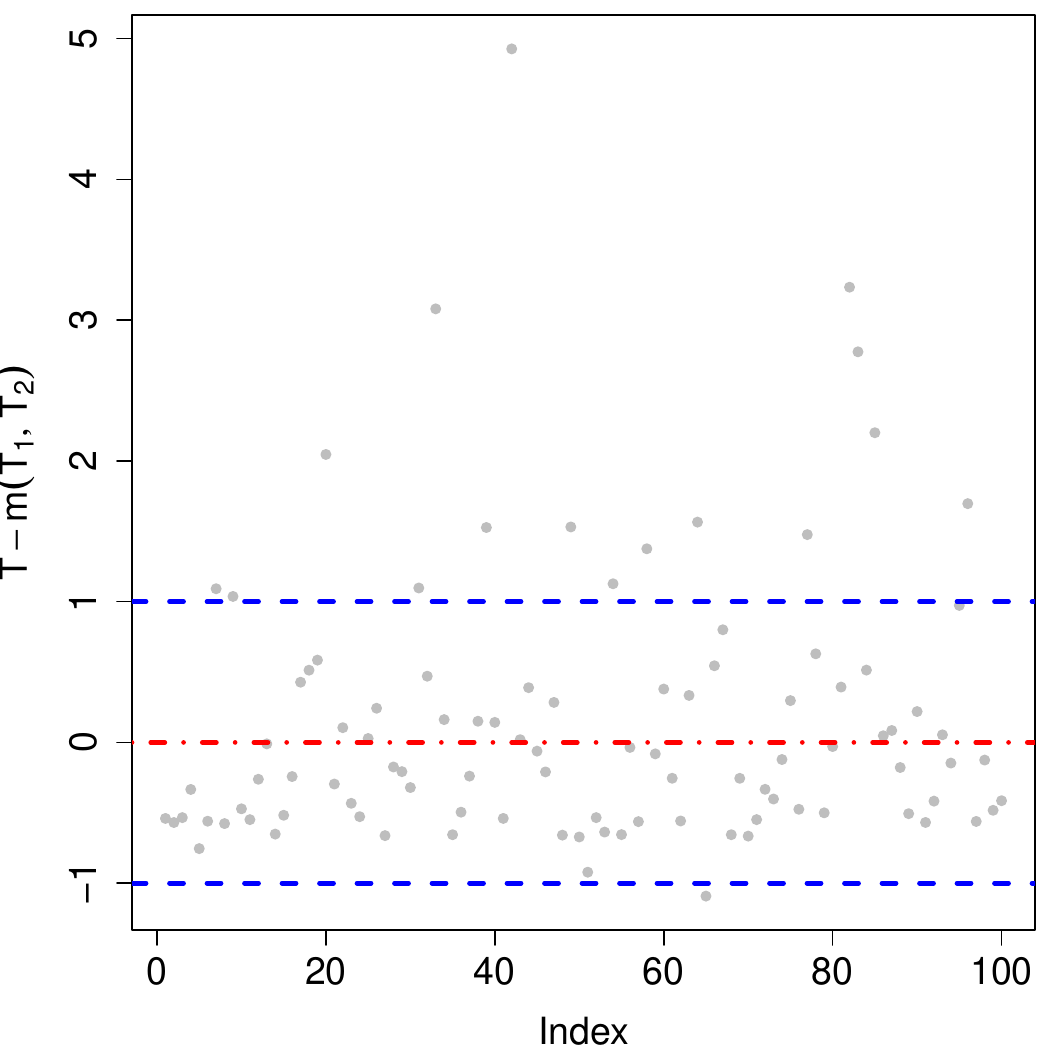}
		\caption{Comparisons (left) between the predictions from the median regression map $m(T_1,T_2)$ \textcolor{blue}{(red-dot-dashed line)} and the exact system lifetime values $T$ (gray points) for the systems in Example \ref{ex6} in a simulated sample. The bias $T-m(T_1,T_2)$ for that predictions are in the right plot.} \label{fig6}
	\end{center}
\end{figure}

For example, in the first data in our simulated sample, the first component failure happens at time $t_1=0.4632196$. At this point, our prediction for $T$ is
$$\bar G^{-1}_{3|1}(0.5|0.4632196)= 1.6585,$$
with the $90\%$ centered prediction interval $[0.7117,4.0781]$. The exact value for $T$ in this first data is $T=0.8434573$ which belongs to this interval. 

The second component failure happens at $t_2=0.6899807$. With this additional information,  the prediction for $T$ is 
$$m(0.4632196,0.6899807)=\bar G^{-1}_{3|1,2}(0.5|0.4632196,0.6899807)
=1.383333,$$ which is slightly better than the first prediction.
The $90\%$ centered prediction interval for this second prediction is $[0.7412945,3.686103]$, which also contains the real value of $T$. Furthermore, note that its length is smaller than the one of the intervals obtained at $t_1$.  The prediction in the IID case was $1.383128$ with the prediction interval $[ 0.7413, 3.6857]$. These values are very similar since the FGM copula with $\theta=1$ represents a weak positive dependence. 

In Figure \ref{fig6} (left), we can see the predictions from $T_1$ and $T_2$ \blue{(red-dot-dashed line)} and the exact values for $T$ \blue{(gray points)} jointly with the associated bias (right). The bias is greater than $1$ in $16$ cases and less than $-1$ just in one case (i.e. in $83$-out-of-$100$ cases the absolute errors are less than one). Hence, our predictions are conservative (when the value of $T$ is far from our predictions, it is usually greater than our predictions). In this sample, the mean absolute error is $0.66015$ and the mean squared error is $0.983492$.

\end{example}

\section{Conclusions}\label{Sec5}

We have provided helpful procedures to predict system failure times from early components' failure times. We have considered the most typical cases: a premature failure that always happens before the system failure (case I), the same assumption with a possible failure of the system at this point (case II), and two early failures previous to the system failure (case III). All the cases are solved using multivariate distortions and quantile regression techniques for independent and dependent components.  Other cases can be solved similarly  using the \blue{representation results provided in \cite{Alfonso,N18,NAS19,NB22, NC19,ND17,NPL17}.} We include illustrative examples by applying these theoretical results to specific cases and showing how we can proceed in practice by using empirical quantile regression tools.   

All the studied cases are based on the assumption of homogeneous  (identically distributed)  components. Therefore, the main task for future research is to extend these results to systems with heterogeneous components. Another further research is to develop parametric inference procedures for specific systems, reliability functions  and copulas to estimate the unknown parameters in the model (i.e., the parameters in the reliability function and/or the survival copula), see e.g. \cite{MN22,YNB16} and the references therein. 

\section*{Acknowledgements}

\blue{We would like to thank the anonymous reviewers for several helpful suggestions that have served to add clarity and breadth to the earlier version of this paper.}
JN thanks the support of Ministerio de Ciencia e Innovaci\'on of Spain under grant PID2019-103971GB-I00/AEI/10.13039/501100011033. AA and AS
thank the partial support of Ministerio de Econom\'ia y Competitividad of Spain under grant PID2020-116216GB-I00, by the 2014-2020 ERDF Operational Programme and by the Department of Economy, Knowledge, Business, and University of the Regional Government of Andalusia, Spain under grant FEDER-UCA18-107519. The authors JN and AA state that this manuscript is part of the project TED2021-129813A-I00 and they thank the support of MCIN/AEI/ 10.13039/501100011033 and the European Union ``NextGenerationEU''/PRTR.


\section*{Declaration of competing interest}

The authors declare that they have no known competing financial interests or personal relationships that could have appeared to influence the work reported in this paper.

\end{document}